\documentclass[prodmode,acmtecs]{acmsmall} 

\acmVolume{9}
\acmNumber{4}
\acmArticle{39}
\acmYear{2010}
\acmMonth{3}

\usepackage{latexsym,amssymb,amsmath}
\usepackage{algorithm}
\usepackage[noend]{algpseudocode}
\usepackage{xspace}
\usepackage{enumerate}
\usepackage[defblank]{paralist}
\usepackage{graphicx}
\usepackage{subcaption}
\graphicspath{{imgs/}}
\usepackage{todonotes}
\usepackage{comment}
\usepackage{tabularx}
\usepackage{thesisSymbols}
\usepackage{booktabs}
\usepackage{multirow}
\usepackage{stmaryrd}

\usepackage{tkz-graph}
\usetikzlibrary{matrix,automata,arrows,calc,fit}
\usetikzlibrary{decorations.pathreplacing,positioning}

\graphicspath{{./images/}}

\newcommand{\myi}{(\emph{i})\xspace}
\newcommand{\myii}{(\emph{ii})\xspace}

\newcommand{\A}{\mathcal{A}} 
\newcommand{\C}{\mathcal{C}}

 \renewcommand{\L}{\mathcal{L}}
\newcommand{\M}{\mathcal{M}} 
 \renewcommand{\P}{\mathcal{P}}
 \newcommand{\R}{\mathcal{R}}
\renewcommand{\S}{\mathcal{S}} 
\newcommand{\U}{\mathcal{U}}

\newcommand{\limp}{\mathbin{\rightarrow}}

\newcommand{\Next}{\raisebox{-0.27ex}{\LARGE$\circ$}}
\newcommand{\Wnext}{\raisebox{-0.27ex}{\LARGE$\bullet$}}
\renewcommand{\Until}{\mathop{\U}}
\newcommand{\Release}{\mathop{\R}}

\newcommand{\true}{\mathit{true}}

\newcommand{\false}{\mathit{false}}
\newcommand{\ttrue}{\mathtt{true}}
\newcommand{\ffalse}{\mathtt{false}}
\newcommand{\Last}{\mathit{last}}

\newcommand{\length}{\mathit{length}}

\newcommand{\nnf}{\mathit{nnf}}

\newcommand{\BOX}[1]{ [#1]}
\newcommand{\DIAM}[1]{\langle #1 \rangle}

\newcommand{\LTL}{{\sc ltl}\xspace}
\newcommand{\LTLf}{{\sc ltl}$_f$\xspace}

\newcommand{\LDL}{{\sc ldl}\xspace}
\newcommand{\LDLf}{{\sc ldl}$_f$\xspace}

\newcommand{\PDL}{{\sc pdl}\xspace}

\newcommand{\AFW}{{\sc afw}\xspace}
\newcommand{\NFA}{{\sc nfa}\xspace}
\newcommand{\DFA}{{\sc dfa}\xspace}

\newcommand{\declare}{{\sc declare}\xspace}

\newcommand{\Nat}{{\rm I\kern-.23em N}}
\newcommand{\Prop}{\P}

\newcommand{\rvass}[2]{\ensuremath{\llbracket #1 =#2\rrbracket}}

\newcommand{\pref}{\mathsf{pref}}
\newcommand{\temptrue}{\mathit{temp\_true}}
\newcommand{\tempfalse}{\mathit{temp\_false}}
\newcommand{\permtrue}{\mathit{perm\_true}}
\newcommand{\permfalse}{\mathit{perm\_false}}
\newcommand{\possgood}{\mathit{poss\_good}}

\newcommand{\necgood}{\mathit{nec\_good}}
\newcommand{\necbad}{\mathit{nec\_bad}}

\renewcommand{\ttrue}{\mathit{tt}}
\renewcommand{\ffalse}{\mathit{ff}}
\newcommand{\Endt}{\mathit{end}}

\newcommand{\ot}{o}

\newcommand{\tr}{\mathit{tr}}

\newcommand{\ra}{\rightarrow}

\newcommand{\biject}{\sim}
\newcommand{\bijectf}[1]{\overset{#1}{\sim}}

\newcommand{\expand}{\textbf{\textit{\texttt{E}}}}
\newcommand{\ttt}{\textbf{\textit{\texttt{T}}}}
\newcommand{\fff}{\textbf{\textit{\texttt{F}}}}

\renewcommand{\set}[1]{\{ #1 \}}

\newcommand{\statesize}{.5 cm}

\tikzstyle{dot}=[
    circle,
    ultra thick,
    minimum width=1mm,
    minimum height=1mm,
    draw,
  ]

\tikzstyle{state}=[circle,draw=black,very thick,minimum width=\statesize]

\tikzstyle{istate}=[state,initial,initial text={\!\!\!\!\!},initial distance=3mm]

\tikzstyle{truecolor}=[fill=green!30]
\tikzstyle{temptruecolor}=[fill=blue!30]
\tikzstyle{falsecolor}=[fill=red!30]
\tikzstyle{tempfalsecolor}=[fill=orange!30]

\tikzstyle{truestate}=[state,truecolor,line width=2pt]
\tikzstyle{temptruestate}=[state,temptruecolor,very thick]
\tikzstyle{falsestate}=[state,falsecolor,very thick,densely dotted]
\tikzstyle{tempfalsestate}=[state,tempfalsecolor,very thick,densely dashed]

\tikzstyle{monitorstate}=[
          very thick,
          draw,
          inner sep=0pt]

\tikzstyle{truemonstate}=[
          monitorstate,
          truecolor,
          label=center:$\permtrue$]
          
\tikzstyle{temptruemonstate}=[
          monitorstate,
          very thick,
          temptruecolor,
          label=center:$\temptrue$]
          
\tikzstyle{falsemonstate}=[
          monitorstate,
          very thick,
          falsecolor,
          label=center:$\permfalse$]
                    
\tikzstyle{tempfalsemonstate}=[
          monitorstate,
          very thick,
          tempfalsecolor,
          label=center:$\tempfalse$]
\definecolor{deepblue}{HTML}{0C3B80}
\definecolor{deepgreen}{HTML}{2EA601}
\definecolor{lightOrange}{HTML}{FFA03C}
\definecolor{darkOrange}{HTML}{F1800A}
\definecolor{lightBlue}{HTML}{0174CD}
\definecolor{greenF}{HTML}{2CBB5C}
\definecolor{cyan}{HTML}{86A6D5}
\definecolor{darkred}{HTML}{8B0000}

\usetikzlibrary{automata,positioning}
\usetikzlibrary{arrows}
\usetikzlibrary{decorations.markings}
\usetikzlibrary{shapes.misc, positioning}
\usetikzlibrary{arrows.meta,decorations.markings}

\def\DTZU {2ex}

\tikzstyle{taskfg} = [
  text=deepblue,
]

\tikzstyle{taskbg} = [
  fill=cyan!50,
]

\tikzstyle{taskline} = [
  draw=deepblue,
]

\tikzstyle{taskstyle} = [
  ultra thick,
  taskfg,
  taskbg,
  taskline
]

\tikzstyle{task} = [
  rectangle,
  minimum width=20mm,
  minimum height=10mm,
  taskstyle
]

\tikzstyle{smalltask} = [
  rectangle,
  minimum width=8mm,
  minimum height=6mm,
  taskstyle,
  very thick
]

\tikzstyle{constraint} = [
  ultra thick,
  taskline,
  taskbg
]

\tikzstyle{response} = [
  constraint,
  *-triangle 60
]

\tikzstyle{precedence} = [
  constraint,
  *triangle 60-
]
\tikzstyle{succession} = [
  constraint,
  *triangle 60-*
]

\tikzstyle{respondedexistence} = [
  constraint,
  *-
]

\tikzstyle{coexistence} = [
  constraint,
  *-*
]

\tikzstyle{negationconstraint} = [
  constraint,
  postaction={decorate,decoration={markings,
   mark=at position .5 with {\arrow[xshift=0.15*\DTZU]{Bar[width=1.5*\DTZU]}},
   mark=at position .5 with {\arrow[xshift=-0.15*\DTZU]{Bar[width=1.5*\DTZU]}}
  }}
]

\tikzstyle{notcoexistence} = [
  negationconstraint,
  coexistence,
]

\tikzstyle{negationsuccession} = [
  constraint,
  succession,
  postaction={decorate,decoration={markings,
   mark=at position .65 with {\arrow[xshift=0.15*\DTZU]{Bar[width=1.5*\DTZU]}},
   mark=at position .65 with {\arrow[xshift=-0.15*\DTZU]{Bar[width=1.5*\DTZU]}}
  }}
]

\tikzstyle{exclusivechoice} = [
  constraint,
  -,
  postaction={decorate,decoration={markings,
   mark=at position .5 with {\arrow[xshift=0.5*\DTZU]{Diamond[width=1*\DTZU]}} 
  }}
]

\tikzstyle{choice} = [
  constraint,
  -,
  postaction={decorate,decoration={markings,
   mark=at position .5 with {\arrow[xshift=0.5*\DTZU]{Diamond[width=1*\DTZU]}},
    mark=at position .5 with {\arrow[xshift=0.28*\DTZU,white,scale=.6]{Diamond[width=1*\DTZU]}} 
  }}
]

\tikzstyle{circ} = [
  solid,
  text=black
]

\tikzstyle{timeline} = [
  -,
  thin,
  ]
  
\tikzstyle{snapshot} = [
  -,
  thin,
  densely dotted,
  ]

\tikzstyle{objnode} = [
  inner sep=2pt,
  circle,
  ultra thick,
]

\tikzstyle{tobj} = [
  objnode,
  taskfg,
  taskbg,
  taskline
]

\tikzstyle{cobj} = [
  objnode,
  classfg,
  classbg,
  classline
]

\tikzstyle{link} = [
  solid,
  -angle 60
]

\begin{document}

\markboth{G. De Giacomo et al.}{Monitoring Constraints and Metaconstraints
  with Temporal Logics on Finite Traces}

\title{Monitoring Constraints and Metaconstraints
 \\ with Temporal Logics on Finite Traces} \author{GIUSEPPE {DE GIACOMO} \affil{Sapienza
    Universit\`a\ di Roma} RICCARDO {DE MASELLIS} \affil{Fondazione
    Bruno Kessler} FABRIZIO MARIA MAGGI \affil{University of Tartu}
  MARCO MONTALI \affil{Free University of Bozen-Bolzan}}

\begin{abstract}
Runtime monitoring is one of the central tasks in the area of operational decision support for business process management. In particular, it helps process executors to check on-the-fly whether a running process instance satisfies business constraints of interest, providing an immediate feedback when deviations occur.
We study runtime monitoring of properties expressed in \LTL on finite traces (\LTLf), and in its extension \LDLf.
\LDLf is a powerful logic that captures all monadic second order logic on finite traces, and that is obtained by combining regular expressions with \LTLf, adopting the syntax of propositional dynamic logic (\PDL). Interestingly, in spite of its greater expressivity, \LDLf has exactly the same computational complexity of \LTLf.
We show that \LDLf is able to declaratively express, in the logic itself, not only the constraints to be monitored, but also the de-facto standard RV-LTL monitors. On the one hand, this enables us to directly employ the standard characterization of \LDLf based on finite-state automata to monitor constraints in a fine-grained way. On the other hand, it provides the basis
    for declaratively expressing sophisticated
metaconstraints that predicate on the monitoring state of other constraints, and to check them by relying on standard logical services instead of ad-hoc algorithms.
In addition, we devise a direct translation of \LDLf formulae into
nondeterministic finite-state automata, avoiding to detour to B\"uchi automata or alternating automata. We then report on how this approach has been effectively implemented using Java to manipulate \LDLf formulae and their corresponding monitors, and the well-known \textsc{ProM} process mining suite as underlying operational decision support infrastructure.
\end{abstract}






\acmformat{Giuseppe {De Giacomo}, Riccardo {De Masellis}, Fabrizio
  Maria Maggi, Marco Montali, 2015. Monitoring Business Constraints and Metaconstraints
  with Temporal Logics on Finite Traces.}

\begin{bottomstuff}
  This work is supported by the National Science Foundation, under
  grant CNS-0435060, grant CCR-0325197 and grant EN-CS-0329609.

  Author's addresses:
  G. {De Giacomo}, Dipartimento di Informatica e Sistemistica,
  Sapienza Universit\`a\ di Roma, Via Ariosto, 25, 00185 Rome, Italy;
  R. {De Masellis}, Fondazione Bruno Kessler, via Santa Croce, 77,
  38122 Trento, Italy;
  F. M. Maggi, University of Tartu, J. Liivi 2, 50409 Tartu, Estonia;
  M. Montali, Free University of Bozen-Bolzano, Piazza Domenicani 3,
  39100 Bolzano, Italy.
\end{bottomstuff}

\maketitle

\section{Introduction}
Runtime monitoring is one of the central tasks to provide \emph{operational decision support} to running business processes \cite{Aal16}. While traditional process mining techniques analyze event data of already completed process instances, operational support lifts process mining to running, live process executions, providing an online feedback that can be used to influence the future continuations of such executions. In this setting, the goal of monitoring is to check on-the-fly whether a running process instance complies with business constraints and rules of interest, promptly detecting deviations \cite{LMM13}. Such indicators can, in turn, be used to compute different monitoring metrics, obtaining a succinct summary about the degree of compliance of a running process instance.

In order to provide provably correct runtime monitoring techniques with a well-defined semantics and a solid formal background, monitoring is typically rooted into the field of \emph{formal verification}, the branch of formal methods aimed at checking whether a system meets some property of interest. Being the system dynamic, properties are typically expressed by making use of temporal logics, that is, modal logics whose modal operators predicate about the evolution of the system along time.



Among all the temporal logics used in verification, Linear-time
Temporal Logic (\LTL) is particularly suited for monitoring, as an
actual system execution is indeed a linear sequence of events. Since the \LTL semantics is given in terms of infinite traces, \LTL monitors analyze the trace of interest by considering it as the prefix of an infinite trace that will continue forever \cite{Bauer2010:LTL}. However, this hypothesis falls short in several contexts, where the usual assumption is that each trace produced by the system is in fact finite. This is often the case in Business Process Management (BPM), where each process instance is expected to eventually reach one of the foreseen ending states of the process \cite{PesV06}. In this setting, a monitored trace has to be considered as the prefix of an unknown, but still finitely long, trace. To handle this type of setting, finite-trace variants of \LTL have been introduced. In this work, we consider in particular the logic \LTLf (\LTL on finite traces), investigated in
detail in \cite{DegVa13}, and at the basis of one of the main
declarative process modeling approaches: \declare \cite{PesV06,MPVC10,MMW11}.

Following \cite{MMW11}, monitoring in \LTLf amounts to
check whether the current execution belongs to the set of
admissible \emph{prefixes} for the traces of a given \LTLf formula
$\varphi$. To achieve such a task, $\varphi$ is usually first translated
into a corresponding finite-state automaton that exactly recognizes all and only those
\emph{finite} traces that satisfy $\varphi$.
Despite the presence of previous operational decision support techniques to monitoring \LTLf constraints over finite traces \cite{MMW11,MWM12}, two main
challenges have not yet been tackled in a systematic way. 

First of
all, several alternative semantics have been proposed to make LTL
suitable for runtime verification, such as the de-facto standard RV
monitor conditions \cite{Bauer2010:LTL}, which interpret \LTL formulae using four distinct truth values that account at once for the current trace and its possible future continuations. Specifically, in the RV-\LTL framework, a formula is associated to a corresponding RV state, which may witness:
\begin{inparaenum}[\it (i)]
\item permanent violation (the formula is currently violated, and the violation cannot be repaired anymore);
\item temporary violation (the formula is currently violated but it is possible to continue the execution in a way that makes the formula satisfied);
\item permanent satisfaction (the formula is currently satisfied and it will stay satisfied no matter how the execution continues);
\item  temporary satisfaction (the formula is currently satisfied but may become violated in the future).
\end{inparaenum}
The main issue is that no comprehensive, formal framework based on finite-state automata is available to
handle such RV states. On the one hand, this is because runtime verification for temporal
logics typically focus on the infinite-trace setting \cite{Bauer2010:LTL}, with the consequence that the corresponding
automata-theoretic techniques detour to B\"uchi automata for building and using the monitors. On
the other hand, the incorporation of such an RV semantics in a finite-trace setting has only been tackled so far
with ad-hoc techniques. This is in particular the case of  \cite{MMW11}, which operationally proposes to ``color'' automata
to support the different RV states, but it does not come with an underlying formal counterpart justifying the correctness of the approach.

A second, fundamental challenge is the incorporation of advanced forms of
monitoring, going beyond what can be expressed with \LTLf. In particular, contemporary monitoring approaches do not systematically account for  \emph{metaconstraints} that predicate on the RV state of other constraints. This is especially important in a monitoring setting, where it is often of interest to consider certain constraints only when specific circumstances arise, such as when other constraints become violated.  For example, metaconstraints provide the basis for monitoring
\emph{compensation constraints}, which can be considered as the temporal version of so-called \emph{contrary-to-duty obligations}
\cite{PrS96} in normative reasoning, that is, obligations that are put in place only when other
obligations have not been fulfilled. While this feature is considered
to be a fundamental compliance monitoring functionality \cite{LMM13}, it is
still an open challenge, without any systematic approach able to
support it at the level of the constraint specification language.

In this article, we attack these two challenges by proposing a formal and operational framework for the monitoring of properties expressed in \LTLf and in its extension \LDLf \cite{DegVa13}.
\LDLf is a powerful logic that completely captures Monadic Second-Order Logic on finite traces, in turn, expressively equivalent to the language of regular expressions. \LDLf does so by combining regular expressions with \LTLf, adopting the syntax of propositional dynamic logic (\PDL). Interestingly, in spite of its greater expressivity, \LDLf has exactly the same computational complexity of \LTLf. At the same time, it provides a balanced integration between the expressiveness of regular expressions, and the declarativeness of \LTLf.

Our first, technical contribution is the formal development, accompanied by a proof-of-concept implementation, of an automata-theoretic framework for monitoring \LTLf and \LDLf constraints using the four truth values of the RV approach. We do this in two steps. In the first step, we devise a direct translation of \LDLf (and hence of \LTLf)
formulae into nondeterministic automata, which avoid the usual detour
to B\"uchi automata. The technique is grounded on
alternating automata (\AFW), but it actually avoids their introduction all
together: in fact, the technique directly produces a standard non-deterministic
finite-state automaton (\NFA), which can then be manipulated using conventional automata techniques (such as determinization and minimization).
In the second step, we show that \LDLf is able to capture, in the logic itself, special formulae that capture all RV monitoring conditions. More specifically, given  an arbitrary \LDLf formula $\varphi$, we show how to construct, for each RV monitor condition, another \LDLf formula that characterizes all and only the traces culminating in a time point where $\varphi$ is associated to that RV state. By studying the so-obtained four \LDLf special formulae, we then describe how to construct a single automaton that, given a trace, outputs the RV state associated to $\varphi$ by that trace. This, in turn, provides for the first time a proof of correctness of the ``colored automata'' approach proposed in \cite{MMW11}.

We exploit this meta-level ability of \LDLf in our second, major contribution, which shows how to use the logic to capture
 \emph{metaconstraints}, and how to monitor them by relying on usual logical services instead of ad-hoc algorithms.
 Metaconstraints provide a well-founded, declarative basis
to specify and monitor constraints depending on the monitoring state of
other constraints. To concretely show the flexibility and sophistication of our approach, we introduce and study three interesting classes of metaconstraints. The first class is about \emph{contextualizing} a constraint, by expressing that it has to be enforced only in those time points where another constraint is in a given RV state. The second class deals with two forms of the aforementioned \emph{compensation constraints}, which capture that a compensating constraint has to be monitored when another constraint becomes permanently violated. The third and last class targets the interesting case of \emph{conflicting constraints}, that is, constraints that, depending on the circumstances, may contradict each other. In particular, we show how to express a \emph{preference} on which constraint should be satisfied when a contradiction arises.

In the final part of the paper, we report on how our monitoring framework has been concretely implemented, and exposed as
 an operational decision support plug-in within \textsc{ProM}, one of the most widely adopted infrastructures for process mining.\footnote{\url{http://www.promtools.org/}}

This article is a largely extended version of the conference paper in \cite{DDGM14}. In relation with \cite{DDGM14}, we expand all technical parts, including here full proofs of the obtained  results and a completely novel part on the construction of ``colored automata'' for monitoring. In addition, we provide here a much more detailed account on metaconstraints, introducing three metaconstraint classes that have not yet been investigated in prior work. We also report here on the complete implementation of our monitoring framework.

The rest of the article is structured as follows. In Section~\ref{sec:LTLf-LDLf}, we introduce syntax and semantics of \LDLf and \LTLf. In Section~\ref{sec:automaton}, we then show how an \LDLf/\LTLf formula can be translated into a corresponding \NFA that accepts all and only the traces that satisfy the formula. In Section~\ref{sec:rtm}, we show how \LDLf is able to capture the RV states in the logic itself, and employ the automata-theoretic approach developed in Section~\ref{sec:automaton} to construct RV monitors for \LDLf/\LTLf formulae. In Section~\ref{sec:monitoringDeclare}, we discuss how the resulting framework can be applied in the context of the \declare constraint-based process modeling approach. In Section~\ref{sec:monitoring-metaconstraints}, we turn to metaconstraints, introducing the three interesting metaconstraint classes of contextualization, compensation, and preference in case of conflict. The implementation of our monitoring framework in Java and \textsc{ProM} is reported in Section~\ref{sec:implementation}. Conclusion follows.



\section{Linear Temporal Logics on Finite Traces}
\label{sec:LTLf-LDLf}

In this work, we adopt the standard \LTL and its extension \LDL, interpreted on finite traces.
\LTL on finite traces, called \LTLf \cite{DegVa13}, has exactly the same syntax as \LTL on infinite traces \cite{Pnueli77}. Namely, given a set of $\Prop$ of propositional symbols, \LTLf formulae are obtained through the following:
\[\varphi ::= \phi \mid \lnot \varphi \mid \varphi_1\land \varphi_2 \mid \varphi_1\lor \varphi_2
 \mid \Next\varphi \mid \Wnext\varphi \mid \varphi_1\Until\varphi_2 \mid \varphi_1\Release\varphi_2\]
 where $\phi$ is a propositional formula over $\Prop$,
 $\Next$ is the \emph{next} operator,
 $\Wnext$ is the \emph{weak next} operator, for which we have the
 equivalence $\Wnext\varphi \equiv \lnot\Next\lnot
 \varphi$ (notice that in the finite trace case
 $\lnot\Next\lnot \varphi \neq \Next \varphi$),
 $\Until$ is the \emph{until} operator and
 $\Release$ is release operator, for which we have the equivalence
 $\varphi_2 \Release \varphi_2\equiv \lnot (\lnot\varphi_2 \Until
 \lnot\varphi_2)$.  In addition, we have common abbreviations. For
 example, \emph{eventually} $\Diamond\varphi$ abbreviates
 $\true \Until \varphi$; and \emph{always}
 $\Box\varphi$ abbreviates $\false \Release \varphi$ or equivalently  $\lnot\Diamond\lnot\varphi$.

Notice that, for convenience and without loss of generality, we allow negation only in propositional formulae, i.e., we essentially assume the temporal formulae to be in negation normal form (NNF). An arbitrary temporal formula can be put in NNF in linear time.

The semantics of \LTLf is given in terms of \emph{finite traces}
denoting finite, \emph{possibly empty}, sequences
$\pi=\pi_0,\ldots,\pi_n$ of elements from the alphabet $2^\Prop$,
containing all possible propositional interpretations of the
propositional symbols in $\Prop$.
We denote the length of the trace $\pi$ as $\length(\pi)\doteq n+1$.
We denote as $\pi(i)\doteq \pi_i$ the $i$-th step in the trace. If the trace is shorter and  does not include an $i$-th step,  $\pi(i)$ is undefined.
We denote by $\pi(i,j)\doteq \pi_i,\pi_{i+1},\ldots,\pi_{j-1}$, the segment of the trace $\pi$ starting at the $i$-th step and ending at the $j$-th step (excluded). If $j > \length(\pi)$ then $\pi(i,j) = \pi(i,\length(\pi))$.
For every $j \le i$, we have $\pi(i,j) = \epsilon$, i.e., the empty trace.
Notice that here,
differently form \cite{DegVa13}, we allow the empty trace $\epsilon$ as in \cite{BrDP18}. This is convenient for composing monitors, as it will become clear later on in the article.
Given a finite trace $\pi$, we inductively define when an
\LTLf formula $\varphi$ \emph{is true} at a step $i$
written $\pi,i\models\varphi$, as follows (we include abbreviations for convenience):
\begin{itemize}
\item $\pi,i\models \phi$ ~iff~ $0\leq i \leq \length(\pi)$ and $\pi(i)\models\phi$ \quad\text{($\phi$ propositional)};
\item $\pi,i\models \lnot \varphi$ ~iff~  $\pi,i\not\models \varphi$;
\item $\pi,i\models \varphi_1\land\varphi_2$ ~iff~ $\pi,i\models\varphi_1$ and
  $\pi,i\models\varphi_2$;
\item $\pi,i\models \varphi_1\lor\varphi_2$ ~iff~ $\pi,i\models\varphi_1$ or
  $\pi,i\models\varphi_2$;
\item $\pi,i\models \Next\varphi$ ~iff~ $0\leq i<\length(\pi)-1$ and $\pi,i{+}1\models\varphi$;
\item $\pi,i\models \Wnext\varphi$ ~iff~ $0\leq i<\length(\pi)-1$ implies $\pi,i{+}1\models\varphi$;
\item $\pi,i\models \Diamond\varphi$ ~iff~ for some $j$ s.t.\ $0\leq i\leq j <\length(\pi)$, we have $\pi,j\models\varphi$;
\item $\pi,i\models \Box\varphi$ ~iff~ for all $j$ s.t.\ $0\leq i\leq j < \length(\pi)$, we have
  $\pi,j\models\varphi$;
\item $\pi,i\models \varphi_1\Until\varphi_2$ ~iff~ for some $j$ s.t.\ $1\leq i\leq j < \length(\pi)$, we have  $\pi,j\models\varphi_2$, and for all $k$, $i\leq k<j$, we have  $\pi,k\models\varphi_1$;
\item $\pi,i\models \varphi_1\Release\varphi_2$ ~iff~ for all $j$ s.t.\ $0\leq i\leq j < \length(\pi)$,  either  we have $\pi,j\models\varphi_2$ or for some  $k$, $i\leq k<j$, we have  $\pi,k\models\varphi_1$.
\end{itemize}

Observe that for $i\ge\length(\pi)$, hence e.g., for $\pi=\epsilon$ we get:
\begin{itemize}
\item $\pi,i\not\models \phi$ \quad\text{($\phi$ propositional)};
\item $\pi,i\models \lnot\varphi$ ~iff~  $\pi,i\not\models\varphi$;
\item $\pi,i\models \varphi_1\land\varphi_2$ ~iff~ $\pi,i\models\varphi_1$ and
  $\pi,i\models\varphi_2$;
\item $\pi,i\models \varphi_1\lor\varphi_2$ ~iff~ $\pi,i\models\varphi_1$ or
  $\pi,i\models\varphi_2$;
\item $\pi,i\not\models \Next\varphi$;
\item $\pi,i\models \Wnext\varphi$;
\item $\pi,i\not\models \Diamond\varphi$;
\item $\pi,i\models \Box\varphi$;
\item $\pi,i\not\models \varphi_1\Until\varphi_2$;
\item $\pi,i\models \varphi_1\Release\varphi_2$.
\end{itemize}

It is known that \LTLf is as expressive as First-Order Logic over finite traces, so strictly less expressive than regular expressions, which, in turn, are as expressive as Monadic Second-Order logic over finite traces.  On the other hand,  regular expressions are a too low
level formalism for expressing temporal specifications, since, for example, they miss a direct construct for negation and for conjunction~\cite{DegVa13}.

To be as expressive as regular expressions and, at the same time,
convenient as a temporal logic, in \cite{DegVa13} \emph{Linear Dynamic
  Logic of Finite Traces}, or \LDLf, has been proposed. This logic is
as natural as \LTLf, but with the full expressive power of Monadic
Second-Order logic over finite traces.  \LDLf is obtained by merging
\LTLf with regular expressions through the syntax of the well-know
logic of programs \PDL, \emph{Propositional Dynamic
  Logic} \cite{FiLa79,HaKT00}, but adopting a semantics based on finite traces.
\LDLf is an adaptation of \LDL introduced in \cite{Var11},
which, like \LTL, is interpreted over infinite traces.

Formally, \LDLf formulae are built as follows:
\[\begin{array}{lcl}
\varphi &::=& \ttrue \mid \ffalse\mid  \lnot\varphi \mid
\varphi_1 \land \varphi_2 \mid  \varphi_1 \lor \varphi_2 \mid
\DIAM{\rho}\varphi \mid \BOX{\rho}\varphi
\\
\rho &::=&\phi \mid \varphi? \mid  \rho_1 + \rho_2 \mid \rho_1; \rho_2 \mid \rho^*
\end{array}
\]
where  $\ttrue$ and
$\ffalse$ denote respectively the true and the false \LDLf formula
(not to be confused with the propositional formula $\true$ and
$\false$); $\phi$ denotes propositional formulae over $\Prop$; $\rho$ denotes path expressions, which are regular expressions
over propositional formulae $\phi$ over $\Prop$ with the addition of the test
construct $\varphi?$ typical of \PDL and are used to insert into the execution path checks for
satisfaction of additional \LDLf formulae; and $\varphi$ stand for \LDLf
formulae built by applying boolean connectives and the modal
operators $\DIAM{\rho}\varphi$ and $\BOX{\rho}\varphi$. These two operators are linked by the following equivalence $\BOX{\rho}\varphi\equiv\lnot\DIAM{\rho}{\lnot\varphi}$.

Intuitively, $\DIAM{\rho}\varphi$ states that, from the current step in the trace,
there exists an execution satisfying the regular expression $\rho$
such that its last step satisfies $\varphi$. While
$\BOX{\rho}\varphi$ states that, from the current step, all
executions satisfying the regular expression $\rho$ are such that
their last step satisfies $\varphi$.

Notice that \LDLf, as defined above, does not include propositional formulae $\phi$ as \LDLf formulae, but only as path expressions. However, they can be immediately introduced as abbreviations: $\phi \doteq \DIAM{\phi}\ttrue$.
For example, to say that eventually proposition $a$
holds, instead of writing $\DIAM{\true^*}a$, we can write
$\DIAM{\true^*;a}\ttrue$. This is analogous to what happens in
(extensions with regular expressions of) XPath, a well-known formalism
developed for navigating XML documents and graph databases
\cite{ClDe99b,Marx04b,CDLV09}.
We may keep $\phi$ as \LDLf formulae for convenience, however, we have to be careful of the difference we get if we apply negation to propositional formula $\phi$ or to $\DIAM{\phi}\ttrue$. In the first case, we get $\lnot\phi$, which is equivalent to $\DIAM{\lnot \phi}\ttrue$. In the second case, we get $\BOX{\phi}\ffalse$, which is equivalent to $\BOX{true?}\ffalse \lor \DIAM{\lnot \phi}\ttrue$, which says that either the trace is empty or $\phi$ holds in the current state. We drop the use of $\phi$ to avoid this ambiguity.

It is also convenient to introduce the following abbreviations
specific for dealing with the finiteness of the traces:
$\Endt=\BOX{\true}\ffalse$, which denotes that the trace has been
completed (the current instant is out of the range of the trace, or
the remaining fragment of the trace is empty); and
$\Last= \DIAM{\true}\Endt$, which denotes the last step of the trace.

As for \LTLf, the semantics of \LDLf is given in terms of \emph{finite
  traces} denoting a finite, \emph{possibly empty}, sequence of
consecutive steps in the trace, i.e., finite words $\pi$ over the
alphabet of $2^\Prop$, containing all possible propositional
interpretations of the propositional symbols in $\Prop$.
 The semantics of \LDLf is given in the following. An
\LDLf formula $\varphi$ \emph{is true} at a step $i$, in symbols
$\pi,i\models\varphi$, if:
\begin{itemize}
\vspace{-0.2cm}
\item $\pi,i\models \ttrue$;
\item $\pi,i\not\models \ffalse$;
\item $\pi,i\models \lnot \varphi$ ~iff~ $\pi,i\not \models\varphi$;
\item $\pi,i\models \varphi_1\land\varphi_2$ ~iff~ $\pi,i\models\varphi_1$ and
  $\pi,i\models\varphi_2$;
\item $\pi,i\models \varphi_1\lor\varphi_2$ ~iff~ $\pi,i\models\varphi_1$ or
  $\pi,i\models\varphi_2$;
\item $\pi,i\models \DIAM{\rho}\varphi$ ~iff~  for some $j$ s.t.\
$i\leq j$, we have $\pi(i,j)\in \L(\rho)$ and $\pi,j\models\varphi$;
\item $\pi,i\models \BOX{\rho}\varphi$ ~iff~  for all $j$ s.t.\
$i\leq j$; 
and $\pi(i,j)\in \L(\rho)$, we have $\pi,j\models\varphi$.
\end{itemize}

The relation $\pi(i,j)\in\L(\rho)$ 
 is defined inductively as follows:
\begin{itemize}
\item $\pi(i,j)\in\L(\phi)$ if $j=i+1~\text{and}\; 0\leq i \leq \length(\pi) \;\text{and}\; \pi(i)\models \phi  \quad \mbox{($\phi$ propositional)}$;
\item $\pi(i,j)\in\L(\varphi?)$ if $j=i\;\text{and}\;  \pi, i\models \varphi$;
\item $\pi(i,j)\in\L(\rho_1+ \rho_2)$ if $\pi(i,j)\in\L(\rho_1)\;\text{or}\; \pi(i,j)\in\L(\rho_2)$;
\item $\pi(i,j)\in\L(\rho_1; \rho_2)$ if  $\mbox{ exists } k \mbox{  s.t.\  } \pi(i,k)\in\L(\rho_1) \mbox{ and } \pi(k,j)\in\L(\rho_2)$;
\item $\pi(i,j)\in\L(\rho^*)$ if $j=i\; \text{or} \mbox{ exists } k \mbox{ s.t.\  } \pi(i,k)\in\L(\rho) \mbox{ and } \pi(k,j)\in\L(\rho^*)$.
\end{itemize}

Note that if $i\ge\length(\pi)$, hence, e.g., for $\pi=\epsilon$, the above definitions still apply;
though,
$\DIAM{\phi}\varphi$ ($\phi$ prop.) and $\DIAM{\psi}\varphi$ become trivially false. As usual, we write $\pi \models \varphi$ as a shortcut for $\pi,0 \models \varphi$.

It easy to encode \LTLf into  \LDLf: we can define a translation function $\tr$ defined by induction of the \LTLf formula as follows:
\[\begin{array}{rcl}
\tr(\phi) &=& \DIAM{\phi}\ttrue \quad\text{($\phi$ propositional)}\\
\tr(\lnot \varphi) &=& \lnot \tr(\varphi) \\
\tr(\varphi_1\land \varphi_2) &=&  \tr(\varphi_1)\land \tr(\varphi_2)\\
\tr(\varphi_1\lor \varphi_2) &=&  \tr(\varphi_1)\lor \tr(\varphi_2)\\
\tr(\Next\varphi) &=&  \DIAM{true}(\tr(\varphi) \land \lnot \Endt) \\
\tr(\Wnext\varphi) &=&  \tr( \lnot(\Next (\lnot \varphi)))\\
\tr(\Diamond \varphi) &=& \DIAM{true^*} (\tr(\varphi) \land \lnot \Endt) \\
\tr(\Box \varphi) &=& \tr (\lnot(\Diamond(\lnot \varphi)) )\\
\tr(\varphi_1\Until \varphi_2) &=&  \DIAM{(\tr(\varphi_1)?;\true)^*}(\tr(\varphi_2) \land \lnot \Endt)\\
\tr(\varphi_1\Release \varphi _2) &=& \tr( \lnot (\lnot \varphi_1 \Until \lnot \varphi_2))
\end{array}
\]
where $\nnf(\psi)$ is the function that transform $\psi$ by pushing negation inside until it is just used in front of atomic propositions.
It is also easy to encode regular expressions, used as a specification formalism for traces into \LDLf: $\rho$ translates to $\DIAM{\rho} \Endt$.

We say that a trace satisfies an \LTLf/\LDLf formula $\varphi$, written $\pi\models \varphi$, if $\pi,0\models \varphi$. Note that if $\pi$ is the empty trace, and hence $0$ is out of range, still the notion of $\pi,0\models \varphi$ is well defined. Also sometimes we denote  by $\L(\varphi)$ the set of  traces that satisfy $\varphi$. i.e., $\L(\varphi) = \{\pi\mid \pi \models \varphi\}$.


\section{\MakeLowercase{\LDLf} Automaton}
\label{sec:automaton}
 We can associate with each \LDLf formula $\varphi$ an (exponential)
\NFA $A(\varphi)$ that accepts exactly the traces that make $\varphi$
true.
Here, we provide a simple direct algorithm for computing the \NFA
corresponding to an \LDLf formula.  The correctness of the algorithm
is based on the fact that \myi we can associate each \LDLf
formula $\varphi$ with a polynomial \emph{alternating automaton on words}
(\AFW) that accepts exactly those traces that make $\varphi$
true \cite{DegVa13}, and \myii every \AFW can be transformed into an
\NFA, see, e.g., \cite{DegVa13}. 
However, to formulate the algorithm, we do not need these notions, but
we can work directly on the \LDLf formula.
Then, we define an auxiliary function $\delta$ as in Figure~\ref{fig:delta}, which takes an \LDLf
formula $\psi$ (in negation normal form) and a propositional
interpretation $\Pi$ for $\Prop$, or a special symbol $\epsilon$, and returns a
positive boolean formula whose atoms are (quoted)
 sub-formulae of $\psi$.
\begin{figure}[t!]
\begin{align*}
\delta(\ttrue,\Pi) & = \true\\
\delta(\ffalse,\Pi) & = \false\\
\delta(\phi,\Pi) &= \delta(\DIAM{\phi}\ttrue,\Pi) \quad \mbox{($\phi$ prop.)}\\
\delta(\varphi_1\land\varphi_2,\Pi) & =
\delta(\varphi_1,\Pi) \land \delta(\varphi_2,\Pi)\\
\delta(\varphi_1\lor\varphi_2,\Pi) & =
\delta(\varphi_1,\Pi) \lor \delta(\varphi_2,\Pi)\\
\delta(\DIAM{\phi}\varphi,\Pi) & =
\left\{\hspace{-1ex}\begin{array}{l}
\expand(\varphi) \mbox{ if } \Pi \models \phi \quad \mbox{($\phi$  prop.)}\\
\false \mbox{ if } \Pi\not\models \phi
\end{array}\right.\\[5pt]
\delta(\DIAM{\psi?}{\varphi},\Pi) & =
\delta(\psi,\Pi) \land \delta(\varphi,\Pi)\\
\delta(\DIAM{\rho_1+\rho_2}{\varphi},\Pi) & =
\delta(\DIAM{\rho_1}\varphi,\Pi)\lor\delta(\DIAM{\rho_2}\varphi,\Pi)\\
\delta(\DIAM{\rho_1;\rho_2}{\varphi},\Pi) & =
\delta(\DIAM{\rho_1}\DIAM{\rho_2}\varphi,\Pi)\\
\delta(\DIAM{\rho^*}\varphi,\Pi) & =
\delta(\varphi,\Pi) \lor
\delta(\DIAM{\rho}\fff_{\DIAM{\rho^*}\varphi},\,\Pi)\\
\delta(\BOX{\phi}\varphi,\Pi) & =
\left\{\hspace{-1ex}\begin{array}{l}
   \expand(\varphi) \mbox{ if }  \Pi \models \phi \quad \mbox{($\phi$ prop.)}\\
    \true \mbox{ if } \Pi\not\models \phi
\end{array}\right.\\[5pt]
\delta(\BOX{\psi?}{\varphi},\Pi) & =
\delta(\nnf(\lnot\psi),\Pi) \lor \delta(\varphi,\Pi)\\
\delta(\BOX{\rho_1+\rho_2}{\varphi},\Pi) & =
\delta(\BOX{\rho_1}\varphi,\Pi)\land\delta(\BOX{\rho_2}\varphi,\Pi)\\
\delta(\BOX{\rho_1;\rho_2}{\varphi},\Pi) & =
\delta(\BOX{\rho_1}\BOX{\rho_2}\varphi,\Pi)\\
\delta(\BOX{\rho^*}\varphi,\Pi) & =
\delta(\varphi,\Pi) \land
\delta(\BOX{\rho}\ttt_{\BOX{\rho^*}\varphi},\,\Pi)\\
\delta(\fff_{\psi},\Pi) & =  \false\\
\delta(\ttt_{\psi},\Pi) & =  \true
\end{align*}
\caption{Definition of $\delta$, where $\expand(\varphi)$ recursively replaces in $\varphi$ all occurrences of atoms of the form $\ttt_\psi$ and $\fff_\psi$ by $\psi$.}\label{fig:delta}
\end{figure}
Note that for defining $\delta$ we make use of extra
symbols of the form $\ttt_{\DIAM{\rho^*}\varphi}$ and
$\fff_{\BOX{\rho^*}\varphi}$, for handling formulae
$\DIAM{\rho^*}\varphi$ and $\BOX{\rho^*}\varphi$. Such extra symbols act in $\delta$ as they were
additional states excepts that during the recursive computation of
$\delta$ they disappear, either because evaluated to $\true$ or
$\false$ or because they are syntactically replaced by
$\DIAM{\rho}\varphi$ and $\BOX{\rho}\varphi$, respectively, when a new
state is returned. For the latter, we use an auxiliary function
$\expand(\varphi)$, which takes as input a formula $\varphi$ with these
extra symbols $\ttt_\psi$ and $\fff_\psi$ used as additional atomic
propositions, and recursively substitutes in it all their occurrences with the
formula $\psi$ itself. Notice also that for $\phi$ propositional, $\delta(\phi,\Pi) = \delta(\DIAM{\phi}\ttrue,\Pi)$, as a consequence of the equivalence $\phi\equiv\DIAM{\phi}\ttrue$.

The auxiliary function $\delta(\varphi,\epsilon)$, i.e., in the case the (remaining fragment of the) trace is empty, is defined exactly as in Figure~\ref{fig:delta} except for the following base cases :
\begin{align*}
\delta(\DIAM{\phi}\varphi,\epsilon) &=\false \quad \mbox{($\phi$  propositional)}\\
 \delta(\BOX{\phi}\varphi,\epsilon)  &= \true  \quad \;\mbox{($\phi$  propositional)}
\end{align*}
Note that $\delta(\varphi,\epsilon)$ is always either $\true$ or $\false$.

\newcommand{\algoname}{\textsc{{\LDLf}2\NFA}}
\renewcommand{\algorithmicrequire}{\textbf{Input:}}
\renewcommand{\algorithmicensure}{\textbf{Output:}}
\algrenewcommand\algorithmicindent{1em}
 \begin{figure}[!t]
\begin{algorithmic}[1]
\State\textbf{algorithm} \algoname \\
\textbf{input} \LDLf formula $\varphi$ \\
\textbf{output} \NFA $\A(\varphi) = (2^\Prop,\S,s_0,\varrho,S_f)$
 \State $s_0  \gets \{\varphi\}$ \Comment{set the initial state}
 \State $S_f \gets \{\emptyset\}$ \Comment{set final states}
 \If{($\delta(\varphi,\epsilon) = \true$)} \Comment{check if initial state is also final}
    		\State $\S_f \gets  \S_f \cup \{s_0\}$
 \EndIf
 \State $\S \gets \{s_0,\emptyset\}$, $\varrho \gets \emptyset$
 \While{($\S$ or $\varrho$ change)} 
	\For{($s\in \S$)}
		\If{($s'\models \bigwedge_{(\psi\in s)} \delta(\psi,\Pi)$} \Comment{add new state and transition}
    			\State $\S \gets \S \cup \{s'\}$
 			\State $\varrho \gets \varrho \cup \{ (s,\Pi,s')\}$	

		\If{($\bigwedge_{(\psi\in s')} \delta(\psi,\epsilon) = \true$)} \Comment{check if new state is also final}
   			\State $S_f \gets \S_f\cup\{s'\}$
    		\EndIf
	        \EndIf
        \EndFor
 \EndWhile
\end{algorithmic}
 \vspace{-.3cm}
 \caption{\NFA construction.}\label{fig:algo}
 \end{figure}
 
Using the auxiliary function $\delta$, we can build the \NFA
$A(\varphi)$ of an \LDLf formula $\varphi$ in a forward fashion as
described in Figure~\ref{fig:algo}, where:
states of $A(\varphi)$ are sets of atoms (recall that each atom is
quoted $\varphi$ sub-formulae) to be interpreted as a conjunction; the
empty conjunction $\emptyset$ stands for $\true$; $\Pi$ is  a propositional interpretation and $q'$ is a set of (quoted) sub-formulae of $\varphi$ that denotes a minimal interpretation such that $q'\models
\bigwedge_{\psi\in q)} \delta(\psi,\Pi)$. Note that
we do not need to get all $q$ such that $q'\models
\bigwedge_{\psi\in q)} \delta(\psi,\Pi)$, but
only the minimal ones.  In addition, trivially we have
$(\emptyset,a,\emptyset)\in\varrho$ for every $a\in\Sigma$.

The algorithm \algoname\ terminates in at most an exponential number of
steps, and generates a set of states $\S$ whose size is at most
exponential in the size of  $\varphi$.
We observe that the algorithm \algoname\ implicitly constructs the \AFW for $\varphi$, and transforms it into a corresponding \NFA. In particular, given an \LDLf formula $\varphi$, its sub-formulae are the states of the \AFW, with initial state the formula itself, and no final states. The auxiliary function $\delta$ grounded on the
  sub-formulae of $\varphi$ becomes the transition function of such an \AFW. This directly leads to the following result.

\begin{theorem}[\cite{DegVa13}]
\label{thm:automataSoundness}
  Let $\varphi$ be an \LDLf formula and $A(\varphi)$ the \NFA
  obtained by applying the algorithm \algoname\ to $\varphi$. Then $\pi\models\varphi \mbox{ iff } \pi\in
  \L(A(\varphi))$ for every finite trace $\pi$.
\end{theorem}



We can check the satisfiability of an \LDLf formula $\varphi$ by checking whether its corresponding \NFA $A(\varphi)$ is nonempty. The same applies for validity and logical implication, which are linearly reducible to satisfiability.
It is easy to see that $A(\varphi)$  can be built on-the-fly, and hence we can check non-emptiness in PSPACE in the size of $\varphi$. Considering that it is known that satisfiability in \LDLf is PSPACE-hard, we can conclude that the proposed construction is optimal with respect to the computational complexity for satisfiability (see~\cite{DegVa13} for details).

\section{Runtime Monitoring}
\label{sec:rtm}
From a high-level perspective, the monitoring problem amounts to
observe an evolving system execution and report the violation or
satisfaction of properties of interest at the earliest possible
time. As the system progresses, its execution trace increases, and, at
each step, the monitor checks whether the trace seen so far conforms to
the properties, by considering that the execution can still continue.
This evolving aspect has a significant impact on the monitoring output:
at each step, indeed, the outcome may have a
degree of uncertainty due to the fact that future executions are yet
unknown.

Several variants of monitoring semantics have been proposed (see
\cite{Bauer2010:LTL} for a survey). In this work, we adopt the
semantics in \cite{MMW11}, which is essentially the finite-trace variant
of the RV semantics in \cite{Bauer2010:LTL}. Interestingly, in our finite-trace setting the RV semantics can be elegantly defined, since both trace prefixes and their continuations are finite.

Given an \LTLf/\LDLf formula $\varphi$, and a current trace $\pi$, the monitor
returns one among the following four \emph{RV states}:
\begin{compactitem}
\item $\temptrue$, meaning that $\pi$ \emph{temporarily satisfies} $\varphi$, i.e.,
  it satisfies $\varphi$, but there is at least one
  possible continuation of $\pi$ that violates $\varphi$;
\item $\tempfalse$, meaning that $\pi$ \emph{temporarily violates} $\varphi$, i.e.,
  $\varphi$ is not satisfied by $\varphi$, but there is at least one
  possible continuation of $\pi$ that does so;
\item $\permtrue$, meaning that $\pi$ \emph{permanently satisfies} $\varphi$, i.e.,  $\varphi$ is satisfied by $\pi$ and it will always be, no matter how $\pi$ is extended;
\item $\permfalse$, meaning that $\pi$ \emph{permanently violates} $\varphi$, i.e., $\varphi$ it is not satisfied by $\pi$ and it will never be,
  no matter how $\pi$ is extended.
\end{compactitem}

Formally, let $\varphi$ be an \LDLf/\LTLf formula, and let $\pi$ be a trace. Then, we define whether $\varphi$ is in RV state $s \in \set{\temptrue,\tempfalse,\true,\false}$ (written $\rvass{\varphi}{\temptrue}$) on trace $\pi$ as follows:
\begin{compactitem}
\item $\pi \models \rvass{\varphi}{\temptrue}$ if $\pi \models \varphi$ and there exists a trace $\pi'$ such that $\pi \pi' \not\models \varphi$, where $\pi \pi'$ denotes the trace obtained by concatenating $\pi$ with $\pi'$;
\item $\pi \models \rvass{\varphi}{\tempfalse}$ if $\pi \not\models \varphi$ and there exists a trace $\pi'$ such that $\pi \pi' \models \varphi$;
\item $\pi \models \rvass{\varphi}{\permtrue}$ if $\pi \models \varphi$ and for every trace $\pi'$, we have $\pi \pi' \models \varphi$;
\item $\pi \models \rvass{\varphi}{\permfalse}$ if $\pi\not\models \varphi$ and for every trace $\pi'$, we have $\pi \pi' \not\models \varphi$.
\end{compactitem}
By inspecting the definition of RV states, it is straightforward to see that a formula $\varphi$ is in one and only one RV state on a trace $\pi$.

The RV states $\temptrue$ and $\tempfalse$ are not definitive: they may change into any other RV state as the system progresses. This reflects the general
unpredictability of how a system execution unfolds. Conversely, the RV states $\permtrue$ and $\permfalse$ are stable since, once outputted, they will not
change anymore. Observe that a stable RV state can be reached in
two different situations:
 \begin{inparaenum}[\it (i)]
 \item when the system execution terminates;
\item when the formula that is being monitored can be fully evaluated
  by observing a partial trace only.
\end{inparaenum}
The first case is indeed trivial, as when the
execution ends, there are no possible future evolutions and hence it
is enough to evaluate the finite (and now complete) trace seen so far
according to the  \LDLf semantics. In the
second case, instead, it is irrelevant whether the systems continues
its execution or not, since some \LDLf properties, such as
eventualities or safety properties, can be fully evaluated as soon as
something happens, e.g., when the eventuality is verified or the
safety requirement is violated.
Notice also that, when a stable state is returned by the monitor, the monitoring
analysis can be stopped.

From a more theoretical viewpoint, given an \LDLf property $\varphi$, the
monitor looks at the trace seen so far, assesses if it is a \emph{prefix} of a full trace
not yet completed, and categorizes it according to its
potential for satisfying or violating $\varphi$ in the future. We call a prefix \emph{possibly good} for an \LDLf formula
$\varphi$, if there exists an extension of it that satisfies $\varphi$.
More precisely,
  given an \LDLf formula $\varphi$, we define the set of \emph{possibly
    good prefixes for $\L(\varphi)$} as the set:
\begin{equation}\label{def:possGood}
\L_{\possgood}(\varphi) = \{\pi
  \mid \text{ there exists } \pi' \text{ such that } \pi \pi' \in \L(\varphi)\}.
\end{equation}
Prefixes for which every possible extension satisfies $\varphi$ are instead
called \emph{necessarily good}. More precisely,
  given an \LDLf formula $\varphi$, we define the set of
  \emph{necessarily good prefixes for $\L(\varphi)$} as the set:
\begin{equation}\label{def:necGood}
  \L_{\necgood}(\varphi) = \{\pi\mid \text{ for every } \pi'  \text{ such that } \pi \pi' \in \L(\varphi)\}.
\end{equation}
The set of \emph{necessarily bad prefixes} $\L_\necbad(\varphi)$ can be
defined analogously as:
\begin{equation}\label{def:necBad}
  \L_{\necbad}(\varphi) = \{\pi\mid \text{ for every } \pi'  \text{ such that } \pi \pi' \not\in \L(\varphi)\}.
\end{equation}
Observe that the necessarily bad prefixes
for $\varphi$ are the necessarily good prefixes for $\neg \varphi$, i.e., $ \L_{\necbad}(\varphi)= \L_{\necgood}(\lnot\varphi)$.


Such language-theoretic notions allow us to capture all the RV states defined before. More precisely, it is immediate to show the following.

\begin{proposition}
\label{thm:RVLDL}
Let $\varphi$ be an \LDLf formula and $\pi$ a trace. Then:
\begin{compactitem}[$\bullet$]
\item $\pi \models \rvass{\varphi}{\temptrue}$ iff $\pi \in \L(\varphi) \setminus \L_\necgood(\varphi)$;
\item $\pi \models \rvass{\varphi}{\tempfalse}$ iff $\pi \in \L(\neg \varphi) \setminus \L_\necbad(\varphi)$;
\item $\pi \models \rvass{\varphi}{\permtrue}$ iff $\pi \in \L_{\necgood}(\varphi)$;
\item $\pi \models \rvass{\varphi}{\permfalse}$ iff $\pi \in \L_{\necbad}(\varphi)$.
\end{compactitem}
\end{proposition}

We close this section by exploiting the above language-theoretic notions to
better understand the relationships that hold over the various kinds of
prefixes.
We start by observing that the set of all finite words over the alphabet $2^{\P}$ is
the union of the language of $\varphi$ and its complement
$\L(\varphi) \cup \L(\neg \varphi)=(2^{\P})^*$.
Also, any language and its complement are disjoint
$\L(\varphi) \cap \L(\neg \varphi) = \emptyset$.
Since from the definition of possibly good prefixes we have
$\L(\varphi) \subseteq \L_{\possgood}(\varphi)$ and $\L(\neg \varphi) \subseteq
\L_{\possgood}(\neg \varphi)$,  we also have that
$\L_{\possgood}(\varphi) \cup \L_{\possgood}(\neg \varphi)=(2^{\P})^*$.
Also, from this definition, it is easy to see that
$\L_{\possgood}(\varphi) \cap \L_{\possgood}(\neg \varphi)$ corresponds to:
\[
 \{ \pi \mid \text{ there exists } \pi' \text{ such that }\pi\pi' \in \L(\varphi)  \text{ and there exists }
\pi''\text{ such that }\pi\pi'' \in \L(\neg \varphi)\}
\]
meaning that the set of possibly good prefixes for $\varphi$ and the set
of possibly good prefixes for $\neg \varphi$ do intersect, and in such
intersection there are paths that can be extended to satisfy $\varphi$, but can
also be extended to satisfy $\neg \varphi$.
It is also easy to see that
$\L(\varphi) = \L_{\possgood}(\varphi) \setminus \L(\neg \varphi).$

Turning to
necessarily good prefixes and necessarily bad prefixes, it is easy to
see that $\L_{\necgood}(\varphi) = \L_{\possgood}(\varphi) \setminus
\L_{\possgood}(\neg \varphi)$, that $\L_{\necbad}(\varphi) =
\L_{\possgood}(\neg \varphi) \setminus \L_{\possgood}(\varphi)$, and
also that $\L_{\necgood}(\varphi) \subseteq \L(\varphi)\;\; \text{and}\;\; \L_{\necgood}(\varphi) \not \subseteq \L(\neg
\varphi)$.

Interestingly, necessarily good, necessarily bad, and possibly good prefixes partition all finite traces. In fact, by directly applying the definitions of necessarily good, necessarily bad, possibly good prefixes of $\L(\varphi)$ and  $\L(\lnot\varphi)$, we obtain the following.
\begin{proposition}\label{prop:lang-partitions}
The set of all traces $(2^\P)^*$ can be partitioned into
\[
\L_{\necgood}(\varphi)
\qquad\quad
\L_{\possgood}(\varphi) \cap \L_{\possgood}(\neg \varphi)
\qquad\quad
\L_{\necbad}(\varphi)
\]
such that
\[
\begin{array}[t]{l}
\L_{\necgood}(\varphi) \cup (\L_{\possgood}(\varphi) \cap \L_{\possgood}(\neg
  \varphi)) \cup \L_{\necbad}(\varphi) = (2^\P)^*\\
\L_{\necgood}(\varphi) \cap (\L_{\possgood}(\varphi) \cap \L_{\possgood}(\neg \varphi)) \cap \L_{\necbad}(\varphi) =
\emptyset.
\end{array}
\]
\end{proposition}

\input{4.1-monitoring-ldlf}
\subsection{Monitoring using Colored Automata}
\label{sec:colored-proof}

We now show that we can merge the four automata for monitoring the four RV truth values into a single automaton, with ``colored'' states.
The idea is grounded on the intuition that the four automata in the previous section ``have the same shape'', and only differ in determining which states are final. It is hence possible to build one automaton only and then mark its states with four different colors, each corresponding to the final states of a specific formula in Theorem~\ref{thm:rv-ltl}, hence each representing one among the four RV truth values.
The intention of using a single automaton for runtime verification is not novel~\cite{MMW11}, but here, for the first time, we provide a formal justification of its correctness.

As a first step, we formally define the notion of \emph{shape equivalence} to capture the intuition that two automata have the same ``shape'', i.e., they have corresponding states and transitions, but possibly differ in their final states.

  Formally, let $A_1 =
  (2^\Prop,\S^1,s^1_{0}, \varrho^1,
  S^1_{f})$ and $A_2 =
  (2^\Prop,\S^2,s^2_{0},\varrho^2,S^2_{f})$
  be two \NFA{s} defined over a set $\Prop$ of propositional symbols. We say that $A_1$ and $A_2$ are \emph{shape equivalent}, written $A_1 \biject A_2$, if there exists
  a bijection $h:\S^1 \ra \S^2$ such that:
  \begin{compactenum}
  \item $h(s^1_{0})=s^2_{0}$;
  \item for each $(s^1_1, \Pi, s^1_2) \in
    \varrho^1$, $(h(s^1_1), \Pi, h(s^2_2)) \in
    \varrho^2$ and
  \item for each $(s^2_1, \Pi, s^2_2) \in \varrho^2$,
    $(h^{-1}(s^2_1), \Pi, h^{-1}(s^2_2)) \in
    \varrho^1$.
  \end{compactenum}
  We write $A_1 \bijectf{h} A_2$ to explicitly indicate the bijection $h$ from $A_1$ to $A_2$ that induces their shape equivalence.

It is easy to see that bijection $h$ preserves the initial states (condition (1)) and transitions (conditions (2) and (3)) but does not require a correspondence between final states.

\begin{lemma}
  Shape equivalence $\sim$ is indeed an equivalence relation.
\end{lemma}
\begin{proof}
  Reflexivity: the identity function trivially satisfies (1)-(3) above.
  Simmetry: let $A \bijectf{h} A$. Given that $h$ is a
  bijection, then $A \bijectf{h^{-1}} A$.
  Transitivity: Let $A_1 \bijectf{h} A_2$ and $A_2
  \bijectf{g} A_3$. Then $A_1 \bijectf{h\;\circ\;g}
  A_3$, where $h \circ g$ is the composition of $h$ and $g$.
\end{proof}

Hence, $\biject$ induce (equivalence) classes of automata with the same shape. Automata for the basic formulae in Theorem~\ref{thm:rv-ltl} belong to the same class.

\begin{lemma}
  For each \LDLf formula $\varphi$, $A(\varphi)$, $A(\neg \varphi)$,
  $A(\DIAM{\pref_{\varphi}}\Endt)$ and $A(\DIAM{\pref_{\neg
      \varphi}}\Endt)$ are in the same equivalence class by $\biject$.
\end{lemma}
\begin{proof}
  From automata theory, $A(\neg \varphi)$ can be obtained from
  $A(\varphi)$ by switching the final states with the non-final
  ones. Hence, the identity $i: \S^{\varphi} \ra \S^{\varphi}$ is such that $A(\varphi) \bijectf{i} A(\neg \varphi)$.
  Moreover, $A(\varphi) \biject A(\DIAM{\pref_{\varphi}}\Endt) \biject
  A(\DIAM{\pref_{\neg \varphi}}\Endt)$ as $A(\DIAM{\pref_{\varphi}}\Endt)$, resp., $A(\DIAM{\pref_{\neg \varphi}}\Endt)$, can be obtained from $A(\varphi)$, resp., $A(\neg \varphi)$, by setting as final states all states from which there exists a non-zero length path to a final state of $A(\varphi)$, resp., $A(\neg \varphi)$, as explained in the proof of Lemma~\ref{lemma:prefRE}. Hence, again, the identity relation $i$ is such that $A(\varphi) \bijectf{i} A(\DIAM{\pref_{\varphi}}\Endt) \bijectf{i}
  A(\DIAM{\pref_{\neg \varphi}}\Endt)$.
\end{proof}

As the last step for proving that automata for the four formulae in Theorem~\ref{thm:rv-ltl} are in the same class, we show that the formula conjunction does not alter shape equivalence, in the following precise sense.

\begin{theorem}
  Let $\varphi_1$, $\varphi_2$, $\psi_1$ and $\psi_2$ be \LDLf
  formulae  so that $A(\varphi_1) \biject A(\psi_1)$ and
  $A(\varphi_2) \biject A(\psi_2)$. Then $A(\varphi_1 \land
  \varphi_2) \biject A(\psi_1 \land \psi_2)$.
\end{theorem}
\begin{proof}
  From the semantics of \LDLf and Theorem~\ref{thm:automataSoundness}, it 
  follows that $A(\varphi_1 \land \varphi_2) \equiv A(\varphi_1) \cap
  A(\varphi_2)$.   
  Recall that states of $A(\varphi_1) \cap
  A(\varphi_2)$ are ordered pairs $(s^{\varphi_1}, s^{\varphi_2}) \in
  \S^{\varphi_1} \times \S^{\varphi_2}$. 
  
Let $h_1$ and $h_2$ be bijections such that $A(\varphi_1) \bijectf{h_1} A(\psi_1)$ and
  $A(\varphi_2) \bijectf{h_2} A(\psi_2)$. We use $h_1$ and $h_2$ to construct a new bijection  $h:
  \S^{\varphi_1} \times \S^{\varphi_2} \ra \S^{\psi_1} \times
  \S^{\psi_2}$ such that $h(s^{\varphi_1},
  s^{\varphi_2})=(h_1(s^{\varphi_1}), h_2(s^{\varphi_2}))$. We show that $h$ satisfies criteria (1)-(3) of shape equivalence, hence inducing
 $A(\varphi_1 \land \varphi_2) \bijectf{h} A(\psi_1
  \land \psi_2)$.
  The starting state $s_0^{\varphi_1 \land \varphi_2}$ of $A(\varphi_1 \land \varphi_2)$ corresponds to
    $(s_0^{\varphi_1},
  s_0^{\varphi_2})$ by definition of $A(\varphi_1) \cap
  A(\varphi_2)$. At the same time, $s_0^{\psi_1 \land \psi_2}=(h_1(s_0^{\varphi_1}),
  h_2(s_0^{\varphi_2}))=(s_0^{\varphi_1}, s_0^{\varphi_2})$ by
  definition of $h$, which proves (1).
  Now consider a transition $((s^{\varphi_1}_1, s^{\varphi_2}_1), \Pi, (s^{\varphi_1}_2,
  s^{\varphi_2}_2))$ in $\varrho^{\varphi_1 \land \varphi_2}$. By construction, this means that there exist transitions $(s^{\varphi_1}_1, \Pi,
  s^{\varphi_1}_2) \in \varrho^{\varphi_1}$ and $(s^{\varphi_2}_1,
  \Pi, s^{\varphi_2}_2) \in \varrho^{\varphi_2}$. Since
  $A(\varphi_1) \bijectf{h_1} A(\psi_1)$ and $A(\varphi_2)
  \bijectf{h_2} A(\psi_2)$, we have that $(h_1(s^{\varphi_1}_1), \Pi,
  h_1(s^{\varphi_1}_2)) \in \varrho^{\psi_1}$ and
  $(h_2(s^{\varphi_2}_1), \Pi, h_2(s^{\varphi_2}_2)) \in
  \varrho^{\psi_2}$. It follows that $((h_1(s^{\varphi_1}_1),
  h_2(s^{\varphi_2}_1)), \Pi, (h_1(s^{\varphi_1}_2),
  h_2(s^{\varphi_2}_2))) \in \varrho^{\psi_1 \land \psi_2}$, which
  proves (2). Condition (3) is proved analogously with $h^{-1}$.
\end{proof}

\begin{corollary}
  Given an \LDLf formula $\varphi$, automata $A(\varphi \land \DIAM{\pref_{\lnot\varphi}}\Endt)$,
  $A(\lnot\varphi \land \DIAM{\pref_{\varphi}}\Endt)$,
  $A(\DIAM{\pref_{\varphi}}\Endt\land \lnot
  \DIAM{\pref_{\lnot\varphi}}\Endt)$ and
  $A(\DIAM{\pref_{\lnot\varphi}}\Endt\land \lnot
  \DIAM{\pref_{\varphi}}\Endt)$ are in the same equivalence class by
  $\biject$.
\end{corollary}

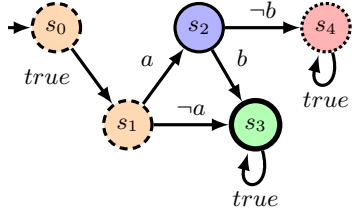
\begin{figure}[h!] 
  \centering
\begin{tikzpicture} [->,>=latex,auto,node distance=1cm, very thick]

\node[tempfalsestate] (s1) {$s_1$};

\node[tempfalsestate,istate] (s0) [above left=1cm of s1, anchor=south] {$s_0$};

\node[temptruestate] (s2) [above right=1 cm of s1, anchor=south] {$s_2$};

\node[truestate] (s3) [right=1 cm of s1, anchor=west] {$s_3$};

\node[falsestate] (s4) [right=1cm of s2, anchor=west] {$s_4$};

\path
(s0) edge node [pos=0.2, below left, align=center] {$true$} (s1)

(s1) edge node [pos=0.5] {$a$} (s2)
      edge node [above] {$\neg a$} (s3)

(s2) edge node [above] {$\neg b$} (s4)
      edge node [pos=0.5] {$b$} (s3)

(s3) edge [loop below] node {$true$} (s3)

(s4) edge [loop below] node [below, align=center] {$true$} (s4);
\end{tikzpicture}
\caption{Graphical representation of the colored automaton for $\Next (a \ra (\Wnext b))$, with the following color coding: automaton states corresponding to the RV state $\temptrue$ are represented in orange, dashed line; those corresponding to $\permtrue$ in green, thick solid line; those corresponding to $\tempfalse$ in blue, thin solid line; those corresponding to $\permfalse$ in red, dotted line. \label{fig:monitoringAutColored}}
\end{figure}

This result tells that the automata of the formulae used to capture the RV states of an \LDLf formula of interest, as captured by Theorem~\ref{thm:rv-ltl}, are identical modulo final states. In addition, by definition of the four \LDLf formulae we directly get that each state is marked as final by one and only one of such automata. This, in turn, allows us to merge all the four automata together into a single automaton, provided that we recall, for each state in the automaton, which of the four formula marks it as final (which corresponds to declare to which of the four RV states it corresponds). In practice, we can simply build the automaton $A(\varphi)$ for $\varphi$, and ``color'' each state in the automaton according to its corresponding RV state. This can be realized with the following, direct procedure. We first build $A(\varphi) =
  (2^\Prop,\S,s_{0}, \varrho,
  S_{f})$ with $S_{f}$ the set of its final states, and for each $s \in \S$ we compute the set $Reach(s)$ of states reachable from $s$. Then:
\begin{compactitem}[$\bullet$]
\item if 
\begin{inparaenum}[\it (i)] 
\item $s \in S_{f}$, 
\item $Reach(s) \not \subseteq S_{f}$, and 
\item $Reach(s) \cap \set{s_{f}} \not = \emptyset$,
 \end{inparaenum}
 then we mark $s$ as $\temptrue$;
\item if \begin{inparaenum}[\it (i)] 
\item $s \not \in S_{f}$,
\item  $Reach(s) \not \subseteq (\S \setminus S_{f})$, and 
\item $Reach(s) \cap S_{f} \not = \emptyset$,
 \end{inparaenum}
 then we mark $s$ as $\tempfalse$;
\item if \begin{inparaenum}[\it (i)] 
\item $s \in S_{f}$ and 
\item $Reach(s) \subseteq S_{f}$, 
 \end{inparaenum}
then we mark $s$ as $\permtrue$;
\item if \begin{inparaenum}[\it (i)] 
\item $s \not \in S_{f}$ and 
\item $Reach(s) \subseteq (\S \setminus S_{f})$, 
 \end{inparaenum}
 then we mark $s$ as $\permfalse$.
\end{compactitem}
It is easy to see that the four bullets above match the four ones of Theorem~\ref{thm:rv-ltl}. The soundness of the marking immediately follows from the definitions and results in the previous section.

\begin{example}
Figure~\ref{fig:monitoringAutColored} depicts the colored automaton for the formula in Example~\ref{ex:monitoringAutomata} and Figure~\ref{subfig:general}. States $s_0$ and $s_1$, which were final in the automaton for $\pi \models \rvass{\varphi}{\tempfalse}$, are indeed marked as $\temptrue$ (orange, dashed line); state $s_2$, which was final in the automaton for $\pi \models \rvass{\varphi}{\temptrue}$, is marked as $\temptrue$ (blue, solid thin line); state $s_3$, which was final in the automaton for $\pi \models \rvass{\varphi}{\permtrue}$, is marked with $\permtrue$ (green, solid thick line) and lastly state $s_4$, which was final in the automaton for $\pi \models \rvass{\varphi}{\permfalse}$ is marked with $\permfalse$ (red, dotted line).
\end{example}

Upon building the colored automaton, we can determinize it and keep it \emph{untrimmed}, that is, including trap states from which no final state can be reached. In this way, every symbol from the entire supporting set $\Prop$ is accepted by the automaton in each of its states. This becomes then our \emph{monitor}.

We conclude by noticing that the presented solution is very flexible, as the reachability analysis can be performed on-the-fly: indeed this is the procedure we actually implemented in our runtime verification tool, as explained in Section~\ref{sec:implementation}.


\renewcommand{\tabularxcolumn}[1]{>{\arraybackslash}m{#1}}

\newcommand{\dm}{\M}
\newcommand{\supp}[2]{\textsc{supp}_{#1}(#2)}
\newcommand{\ftt}[1]{\textsc{to-ft}(#1)}
\newcommand{\itt}[1]{\overline{#1}}
\newcommand{\ft}[1]{\textsc{to-\LTLf}(#1)}
\newcommand{\unique}[1]{\xi_{#1}}
\newcommand{\modelsi}{\models_{\text{LTL}}}

\newcommand{\respondedexistencepayacc}{\ensuremath{\Diamond\activity{pay} \limp \Diamond\activity{acc}}}

\newcommand{\responsepayget}{\ensuremath{\Box (\activity{pay} \limp \Next \Diamond \activity{get})}}

\newcommand{\precedencepayget}{\ensuremath{(\neg \activity{get} \Until \activity{pay})\lor  \neg \Diamond \activity{pay}}}

\newcommand{\absencepay}{\ensuremath{\neg\Diamond (\activity{pay} \land \Next\Diamond \activity{pay})}}

\newcommand{\notcoexistencegetcancel}{\ensuremath{\neg(\Diamond \activity{get} \land \Diamond \activity{cancel})}}

\section{Monitoring Declare Constraints}
\label{sec:declare}

\label{sec:monitoringDeclare}
We now ground our monitoring approach to the case of \declare monitoring.
\declare\footnote{\url{http://www.win.tue.nl/declare/}} is a language and framework for the declarative,
constraint-based modeling of 
processes and
services.
A thorough
treatment of constraint-based processes can be found in \cite{Pes08,Mon10}.
As a modeling language, \declare takes a complementary approach
 to that of classical, imperative process modeling. In imperative process modeling, all allowed control flows
among tasks must be explicitly
represented, and execution traces not falling within this set are implicitly considered as
forbidden.
Instead of this procedural and ``closed'' approach,
\declare has a declarative, ``open'' flavor: the agents responsible
for the process execution can freely choose in which order to perform the
involved tasks, provided that the resulting execution trace satisfies the business constraints of interest.
  This is the reason why, alongside
 traditional control-flow constraints such as sequence (called in
 \declare \emph{chain succession}), \declare supports a variety of more refined
 constraints that impose loose temporal orderings, and/or that
 explicitly account for negative information, i.e., the explicit prohibition of
 task execution.

\newcommand{\nodedist}{6mm}
\newcommand{\taskdist}{9mm}
\newcommand{\autshift}{5mm}

\begin{table}[h!]
\caption{Some \declare constraints, including graphical representation, \LTLf encoding, good prefix characterization, and colored automaton. In the table: $\ot$ is a shortcut notation for \emph{other} task, that is, any task that is not explicitly mentioned by the constraint; given a task $\activity{a}$, $\neg \activity{a}$ is a shortcut notation for any task different from $\activity{a}$.\label{tab:constraints}}
\centering
\begin{tikzpicture}
  \matrix[  nodes={node distance=\nodedist},
            rectangle,draw,
            nodes in empty cells,
            row sep=.5mm,column sep=3mm,
            very thick,
            column 1/.style={anchor=west},
            column 2/.style={anchor=west,xshift=1mm,font=\footnotesize},
            column 3/.style={anchor=west,xshift=1mm,font=\footnotesize},
            column 4/.style={anchor=west},
            >=latex,->,
          ] (declarematrix) {
    \node[yshift=.5mm] {\textsc{constraint}};
    &
    \node {\LTLf};
    &
    \node {$\pref$};
    &
    \node {\textsc{colored automaton}};\\
    \node[smalltask,xshift=1.5mm] (a) {\activity{a}};
    \node[above=-1mm of a,xshift=-.3mm]{\footnotesize\activity{1..$^\ast$}};
    \node[below=0mm of a,font=\footnotesize]{\constraint{existence}};
    &
    \node{$\Diamond \activity{a}$};
    &
    \node {$\true^*$};
    &
    \node[tempfalsestate,xshift=\autshift,istate,initial where=above]  (s0) {$s_0$};
    \node[truestate,right=2*\nodedist of s0]      (s1) {$s_1$};
    \path (s0) edge [loop left] node[below,yshift=-1mm] {$\ot$} (s0);
    \path (s0) edge node[above] {$\activity{a}$} (s1);
    \path (s1) edge [loop right] node[below,yshift=-1mm] {$\true$} (s1);
    \\
    \node[smalltask,xshift=1.5mm] (a) {\activity{a}};
    \node[above=-1mm of a,xshift=-.3mm]{\footnotesize\activity{0}};
    \node[below=0mm of a,font=\footnotesize]{\constraint{absence}};
    &
    \node{$\neg \Diamond \activity{a}$};
    &
    \node {$\ot^*$};
    &
\node[temptruestate,xshift=\autshift,istate,initial where=above]  (s0) {$s_0$};
    \node[falsestate,right=2*\nodedist of s0]  (s1) {$s_1$};
    \path (s0) edge [loop left] node[below,yshift=-1mm] {$\ot$} (s0);
    \path (s0) edge node[above] {$\activity{a}$} (s1);
    \path (s1) edge [loop right] node[below,yshift=-1mm] {$\true$} (s1);
    \\
    \node[smalltask,xshift=1.5mm] (a) {\activity{a}};
    \node[above=-1mm of a,,xshift=-.3mm]{\footnotesize\activity{0..1}};
    \node[below=0mm of a,font=\footnotesize]{\constraint{absence 2}};
    &
    \node{$ \neg \Diamond(\activity{a} \land \Next\Diamond \activity{a})$};
    &
    \node {$\ot^*+(\ot^*;\activity{a};\ot^*)$};
    &
    \node[temptruestate,xshift=\autshift,istate,initial where=above]  (s0) {$s_0$};
    \node[temptruestate,right=2*\nodedist of s0]  (s1) {$s_1$};
    \node[falsestate,below right=of s0,yshift=1mm](s2) {$s_2$};
    \path (s0) edge [loop left] node[below,yshift=-1mm] {$\ot$} (s0);
    \path (s0) edge node[above] {$\activity{a}$} (s1);
    \path (s1) edge [loop right] node[below,yshift=-1mm] {$\ot$} (s1);
    \path[out=-90,in=0] (s1) edge [right] node {$\activity{a}$} (s2);
    \path (s2) edge [loop left] node {$\true$} (s2);
    \\
    \node[smalltask,xshift=1.5mm] (a) {\activity{a}};
    \node[smalltask,right=\taskdist of a] (b) {\activity{b}};
    \path[choice,very thick] (a) -- (b);
    \node[below=-1mm of a.south east,xshift=.5*\taskdist,font=\footnotesize]{\constraint{choice}};
    &
    \node{$\Diamond (\activity{a} \lor \activity{b})$};
    &
    \node {$\true^*$};
    &
    \node[tempfalsestate,xshift=\autshift,istate,initial where=above]  (s0) {$s_0$};
    \node[truestate,right=2*\nodedist of s0]      (s1) {$s_1$};
    \path (s0) edge[loop left] node[below,yshift=-1mm] {$\ot$} (s0);
    \path (s0) edge[out=10,in=170] node[above] {$\activity{a}$} (s1);
    \path (s0) edge[out=-10,in=-170] node[below] {$\activity{b}$} (s1);
    \path (s1) edge[loop right] node[below,yshift=-1mm] {$\true$} (s1);
    \\
    \node[smalltask,xshift=1.5mm] (a) {\activity{a}};
    \node[smalltask,right=\taskdist of a] (b) {\activity{b}};
    \path[respondedexistence,very thick] (a) -- (b);
    \node[below=-1mm of a.south east,xshift=.5*\taskdist,font=\footnotesize]{
      \begin{tabular}{@{}c@{}}
         \constraint{responded}\\\constraint{existence}
      \end{tabular}
    };

    &
    \node{$\Diamond \activity{a} \limp \Diamond \activity{b}$};
    &
    \node {$\true^*$};
    &
    \node[temptruestate,xshift=\autshift,istate,initial where=above]  (s0) {$s_0$};
    \node[tempfalsestate,right=2*\nodedist of s0]  (s1) {$s_1$};
    \node[truestate,below right=of s0](s2) {$s_2$};
    \path (s0) edge [loop left] node[below,yshift=-1mm] {$\ot$} (s0);
    \path (s0) edge node[above] {$\activity{a}$} (s1);
    \path[out=-90,in=180] (s0) edge node[left] {$\activity{b}$} (s2);
    \path (s1) edge [loop right] node[below,yshift=-1mm] {$\neg \activity{b}$} (s1);
    \path[out=-90,in=0] (s1) edge [right] node {$\activity{b}$} (s2);
    \path (s2) edge [loop above] node[right,xshift=1mm,yshift=-3mm] {$\true$} (s2);
    \\
    \node[smalltask,xshift=1.5mm] (a) {\activity{a}};
    \node[smalltask,right=\taskdist of a] (b) {\activity{b}};
    \path[response,very thick] (a) -- (b);
     \node[below=-1mm of a.south east,xshift=.5*\taskdist,font=\footnotesize]{
         \constraint{response}
      };
    &
    \node{$\Box(\activity{a} \limp \Next\Diamond \activity{b})$};
    &
    \node {$\true^*$};
    &
    \node[temptruestate,xshift=\autshift,istate,initial where=above]  (s0) {$s_0$};
    \node[tempfalsestate,right=2*\nodedist of s0]      (s1) {$s_1$};
    \path (s0) edge[loop left] node[below,yshift=-1mm] {$\neg \activity{a}$} (s0);
    \path (s0) edge[out=10,in=170] node[above] {$\activity{a}$} (s1);
    \path (s1) edge[out=-170,in=-10] node[below] {$\activity{b}$} (s0);
    \path (s1) edge[loop right] node[below,yshift=-1mm] {$\neg \activity{b}$} (s1);
    \\
    \node[smalltask,xshift=1.5mm] (a) {\activity{a}};
    \node[smalltask,right=\taskdist of a] (b) {\activity{b}};
    \path[precedence,very thick] (b) -- (a);
    \node[below=-1mm of a.south east,xshift=.5*\taskdist,font=\footnotesize]{
         \constraint{precedence}
      };
    &
    \node{$\neg \activity{b} \Until \activity{a} \lor \neg \Diamond \activity{b}$};
    &
    \node {$(\neg \activity{b})^*+(\ot^*;\activity{a};\true^*)$};
    &
    \node[temptruestate,xshift=\autshift,istate,initial where=above]  (s0) {$s_0$};
    \node[truestate,right=2*\nodedist of s0]  (s1) {$s_1$};
    \node[falsestate,below right=of s0,yshift=1mm](s2) {$s_2$};
    \path (s0) edge [loop left] node[below,yshift=-1mm] {$\ot$} (s0);
    \path (s0) edge node[above] {$\activity{a}$} (s1);
    \path (s1) edge [loop right] node[below,yshift=-1mm] {$\true$} (s1);
    \path[out=-90,in=180] (s0) edge [left] node {$\activity{b}$} (s2);
    \path (s2) edge [loop right] node {$\true$} (s2);
    \\
    \node[smalltask,xshift=1.5mm] (a) {\activity{a}};
    \node[smalltask,right=\taskdist of a] (b) {\activity{b}};
    \path[notcoexistence,very thick] (a) -- (b);
     \node[below=-1mm of a.south east,xshift=.5*\taskdist,font=\footnotesize]{
      \begin{tabular}{@{}c@{}}
         \constraint{not}\\\constraint{coexistence}
      \end{tabular}
    };
    &
    \node{$\neg (\Diamond \activity{a} \land \Diamond \activity{b})$};
    &
    \node {$(\activity{a}+\ot)^*+(\activity{b}+o)^*$};
    &
    \node[temptruestate,xshift=\autshift,istate,initial where=above]  (s0) {$s_0$};
    \node[temptruestate,right=of s0]      (s1) {$s_1$};
    \node[temptruestate,below=of s1]      (s2) {$s_2$};
    \node[falsestate,right=of s1]         (s3) {$s_3$};
    \path (s0) edge[loop left] node[below,yshift=-1mm] {$\ot$} (s0);
    \path (s0) edge node[above] {$\activity{a}$} (s1);
    \path (s0) edge[out=-90,in=180] node[left] {$\activity{b}$} (s2);
    \path (s1) edge[loop above] node[right,xshift=1mm,yshift=-1mm] {$\neg \activity{b}$} (s1);
    \path (s2) edge[loop above] node[right,xshift=1mm,yshift=-3mm] {$\neg \activity{a}$} (s2);
    \path (s1) edge node[above] {$\activity{b}$} (s3);
    \path (s2) edge[out=0,in=-90] node[right] {$\activity{a}$} (s3);
    \path (s3) edge[loop above] node[above,yshift=0mm] {$\true$} (s3);
    \\
  };
  \draw[very thick]
    ($(declarematrix.north west)-(0,8mm)$)
    --
    ($(declarematrix.north east)-(0,8mm)$);

  \draw[very thick]
    ($(declarematrix.north west)-(0,24mm)$)
    --
    ($(declarematrix.north east)-(0,24mm)$);

  \draw[very thick]
    ($(declarematrix.north west)-(0,42mm)$)
    --
    ($(declarematrix.north east)-(0,42mm)$);

    \draw[very thick]
    ($(declarematrix.north west)-(0,65mm)$)
    --
    ($(declarematrix.north east)-(0,65mm)$);

    \draw[very thick]
    ($(declarematrix.north west)-(0,81mm)$)
    --
    ($(declarematrix.north east)-(0,81mm)$);

    \draw[very thick]
    ($(declarematrix.north west)-(0,104mm)$)
    --
    ($(declarematrix.north east)-(0,104mm)$);

    \draw[very thick]
    ($(declarematrix.north west)-(0,120mm)$)
    --
    ($(declarematrix.north east)-(0,120mm)$);

    \draw[very thick]
    ($(declarematrix.north west)-(0,142mm)$)
    --
    ($(declarematrix.north east)-(0,142mm)$);

\end{tikzpicture}
\end{table}

Given a set $\Prop$ of tasks, a \declare model $\M$ is a set $\C$ of \LTLf (and hence \LDLf) constraints over $\Prop$. A finite trace over $\Prop$ \emph{complies with} $\M$, if it satisfies all constraints in $\C$.
Among all possible \LTLf constraints, some specific \emph{patterns} have been
singled out as particularly meaningful for expressing \declare
processes, taking inspiration from \cite{DwAC99}. Such patterns are grouped into four
families:
\begin{compactitem}
\item \emph{existence} (unary) constraints, stating that the target
  task must/cannot be executed (for an indicated amount of times);
\item \emph{choice} (binary) constraints, accounting for alternative tasks;
\item \emph{relation} (binary) constraints, connecting a source task to a target task and expressing that, whenever the
  source task is executed, then the target task must also be executed
  (possibly with additional temporal conditions);
\item \emph{negation} (binary) constraints, capturing that whenever the
  source task is executed, then the target task is prohibited
  (possibly with additional temporal conditions).
\end{compactitem}
Table~\ref{tab:constraints}, at the end of this document, summarizes some of these
patterns. See \cite{MPVC10} for the full list of patterns.

\begin{example}
\label{ex:declare}
Consider a fragment of a ticket booking process, whose \declare representation is shown in Figure~\ref{fig:declare}. The process fragment consists of four tasks and four constraints, but in spite of its simplicity clearly illustrates the main features of declarative, constraint-based process modeling.

Specifically, each process instance is focused on the management of a specific registration to a booking event made by an interested customer. For simplicity, we assume that the type of registration is selected upon instantiating the process, and is therefore not explicitly captured as a set of tasks within the process itself.
The process fragment then consists of four tasks:
\begin{compactitem}[$\bullet$]
\item $\activity{accept regulation}$ is the task used to accept the regulation of the booking company for the specific type of registration the customer is interested in;
\item $\activity{pay registration}$ is the task used to pay for the registration;
\item $\activity{get ticket}$ is the task used to physically withdraw the ticket containing the registration details;
\item $\activity{cancel registration}$ is the task used to abort the instance of the registration process.
\end{compactitem}

The execution of the aforementioned tasks is subject to the following behavioral constraints.
First of all, within an instance of the booking process, a customer may pay for the  registration at most once. This is captured in \declare by constraining the $\activity{pay registration}$ task with an \constraint{absence 2} constraint.

After executing the payment, the customer must eventually get the corresponding ticket. On the other hand, the ticket can be obtained only after having performed the payment. This is captured in \declare by constraining $\activity{pay registration}$ and $\activity{get ticket}$ with a \constraint{response} constraint going from the first task to the second, and with a \constraint{precedence} constraint going from the second task to the first. This specific combination is called \constraint{succession}; it is graphically depicted by combining the graphical notation of the two constraints, and logically corresponds to their conjunction.

When a payment is executed, the customer must accept the regulation of the registration. There is no particular temporal order required for accepting the registration: upon the payment, if the regulation has been already accepted, then no further steps are required; otherwise, the customer is expected to accept the regulation afterwards. This is captured in \declare by connecting  $\activity{pay registration}$ to  $\activity{accept regulation}$ by means of a \constraint{responded existence} constraint.

Finally, a customer may always decide to cancel the registration, with the only constraint that the cancelation is incompatible with the possibility of getting the registration ticket. This means that, when the ticket is withdrawn, no cancelation is accepted anymore, and, on the other hand, once the registration is canceled, no ticket will be issued anymore. This is captured in \declare by relating $\activity{get ticket}$ to $\activity{cancel registration}$ through a \constraint{not coexistence} constraint.

\end{example}

\renewcommand{\taskdist}{2.5cm}

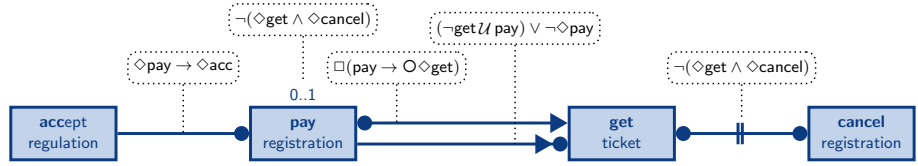
\begin{figure}[t]
\centering
\resizebox{\textwidth}{!} {
\begin{tikzpicture}[node distance=\taskdist]
  \node[task] (r) {
    \begin{tabular}{@{}c@{}}
        \begin{tabular}{@{}c@{}}
        \activity{\textbf{acc}ept}\\
        \activity{regulation}\\
    \end{tabular}
    \end{tabular}
  };
  \node[task, right=of r] (pr) {
    \begin{tabular}{@{}c@{}}
        \activity{\textbf{pay}}\\
        \activity{registration}\\
    \end{tabular}
  };
  \node[task, right=1.6*\taskdist of pr] (wt) {
    \begin{tabular}{@{}c@{}}
        \activity{\textbf{get}}\\
        \activity{ticket}\\
    \end{tabular}
  };

  \node[task, right=of wt] (cr) {
    \begin{tabular}{@{}c@{}}
        \activity{\textbf{cancel}}\\
        \activity{registration}\\
    \end{tabular}
  };

   \node[above=0mm of pr,anchor=south,taskfg] (absence2) {\activity{0..1}};

  \draw[respondedexistence] (pr) -- (r);
  \draw[response] ($(pr.east)+(0,2mm)$) -- ($(wt.west)+(0,2mm)$);
  \draw[precedence] ($(wt.west)-(0,2mm)$) -- ($(pr.east)-(0,2mm)$);
  \draw[notcoexistence] (wt) -- (cr);

  \node[  right=.5*\taskdist of r,
          anchor=center,
          yshift=.5*\taskdist,
          draw,
          dotted,
          rounded corners=5pt,
          thick,
          minimum height=6mm]
      (rex) {
      \respondedexistencepayacc
  };
  \path[thick, dotted]
    (rex)
    edge
    ($(r.east)+(.5*\taskdist,0)$);

   \node[ above=12mm of absence2,
          anchor=center,
          draw,
          dotted,
          rounded corners=5pt,
          thick,
          minimum height=6mm]
          (abs) {
      \notcoexistencegetcancel
  };
  \path[thick,dotted]
    (abs)
    edge
    (absence2);

 \node[  right=.3*\taskdist of pr,
          anchor=center,
          yshift=.5*\taskdist,
          anchor=center,
          draw,
          dotted,
          rounded corners=5pt,
          thick,
          minimum height=6mm]  (resp) {
      \responsepayget
  };
   \path[thick,dotted]
    (resp)
    edge
    ($(pr.east)+(.3*\taskdist,+2mm)$);

  \node[  right=1.2*\taskdist of pr,
          anchor=center,
          yshift=.8*\taskdist,
          anchor=center,
          draw,
          dotted,
          rounded corners=5pt,
          thick,
          minimum height=6mm]  (pre) {
      \precedencepayget
  };
  \path[thick,dotted]
    (pre)
    edge
    ($(pr.east)+(1.2*\taskdist,-2mm)$);

  \node[  right=.5*\taskdist of wt,
          anchor=center,
          yshift=.5*\taskdist,
          anchor=center,
          draw,
          dotted,
          rounded corners=5pt,
          thick,
          minimum height=6mm]
          (nco) {
      \notcoexistencegetcancel
  };
  \path[thick,dotted]
    (nco)
    edge
    ($(wt.east)+(.5*\taskdist,3mm)$);

\end{tikzpicture}
}
\caption{Fragment of a booking process in \declare, showing also the \LTLf formalization of the constraints used therein.
\label{fig:declare}}
\end{figure}

Several, logic-based techniques have been proposed to support end-users in defining, checking, and enacting \declare models \cite{PeSV07,Pes08,Mon09,Mon10,MPVC10}. More recently, the \LTLf characterization of \declare, together with its operational automata-theoretic counterpart, have been exploited to provide advanced monitoring and runtime verification facilities \cite{MMW11,MWM12}. In particular, monitoring \declare models amounts to:
\begin{compactitem}[$\bullet$]
\item Track the evolution of a single \declare constraint against an evolving trace, providing a fine-grained feedback on how the truth value of the constraint evolves, when tasks are performed. This is done by adopting the RV semantics for \LTLf. Specifically, in \cite{MMW11}, the evolution of \declare constraints through the different RV states is tackled using the ad-hoc ``colored automaton" construction technique that we have formally justified in Section~\ref{sec:colored-proof}.
\item Track the compliance of an evolving trace to the entire \declare model, by considering all
  its constraints together. This is done by constructing the colored automaton for the conjunction of all constraints in the model. Monitoring the evolving trace against such a ``global" automaton is crucial for inferring complex violations that cannot be ascribed to the interaction of the current trace with a single constraint in the model, but arise due to the interplay between the trace and multiple constraints at once. Such violations emerge due to \emph{conflicting constraints}, i.e., constraints that, in the current circumstances, contradict each other and consequently cannot be all satisfied anymore \cite{MWM12}. By considering all constraints together, the presence of this kind of conflict can be detected immediately, without waiting for the later moment when an explicit violation of one of the single constraints involved in the conflict eventually arises. This important feature has been classified as \emph{early detection of violations} in a reference monitoring survey \cite{LMM13}.
\end{compactitem}

%
%

\smallskip
\noindent{\bf Monitoring Declare Constraints with \LDLf.}
Since \LDLf includes \LTLf, \declare constraints can be directly encoded in
\LDLf using their standard formalization \cite{PesV06,MPVC10}.
Thanks
to the translation into {\NFA}s discussed in
Section~\ref{sec:automaton} (and, if needed, their determinization into
corresponding {\DFA}s), the automaton obtained from the \LTLf encoding of a constraint can then be used to check whether a (partial) finite trace satisfies that constraint or
not. This is not very effective, as the approach does not support
the detection of fine-grained truth values, as the four RV ones.

By exploiting Theorem~\ref{thm:rv-ltl}, however, we can reuse the same
technique, this time supporting all RV truth values. In fact, by formalizing
the good prefixes of each \declare pattern, we can immediately
construct the four \LDLf formulae that embed the different RV
truth values, and check the current trace over each of the
corresponding automata. Table~\ref{tab:constraints} reports the good
prefix characterization of some of the \declare patterns; it can be
seamlessly extended to all other patterns as well.

Pragmatically, we can even go a step further, and employ colored automata by following the technique discussed in Section~\ref{sec:colored-proof}. More specifically, given a \declare model $\M$, we proceed as follows:
\begin{compactitem}[$\bullet$]
\item For every constraint $\constraint{c} \in \M$, we derive its \LTLf formula $\varphi_\constraint{c}$, and construct its corresponding deterministic colored automaton $A(\constraint{c})$. This colored automaton, acts as \emph{local monitor} for its constraint $\constraint{c}$. This can be used to track the RV state of $\constraint{c}$ as tasks are executed.
\item We build the \LTLf formula $\Phi_{\M}$ standing for the conjunction of the \LTLf formulae encoding all constraints in $\M$, and construct its corresponding deterministic colored automaton $A(\M)$. This colored automaton acts as \emph{global monitor} for the entire \declare model $\M$. This can be used to track the overall RV state of $\M$ as tasks are executed, and early detect violations arising from conflicting constraints.
\item When the monitoring of a process execution starts, the initial state of each local monitor, as well as that of the global monitor, are outputted.
\item Whenever an event witnessing the execution of a task is tracked, it is delivered to each local monitor and to the global monitor. The new, current state of each monitor is then computed and outputted based on the current state and on the received task name.
\item When the process execution is completed (i.e., not further events are expected to occur), the final state of each monitor is outputted, depending on whether its colored automaton is in an accepting state or not. In particular, if upon completion the colored state of the monitor is $\permtrue$ or $\temptrue$, then the trace is judged as compliant; if, instead, upon completion the colored state of the monitor is $\permfalse$ or $\tempfalse$, then the trace is judged as non-compliant.
\item The global monitor can be inquired to obtain additional information about how the monitored trace interacts with the constraints. For example, when the current state of the global monitor is $\tempfalse$ or $\temptrue$, retrieving the names of tasks whose execution leads to a $\permfalse$ state is useful to return which tasks are currently forbidden by the model. This information is irrelevant when the monitor is in a $\permtrue$ or $\permfalse$ state: by definition, in the first case no task is forbidden, whereas in the latter all tasks are.
\end{compactitem}

\renewcommand{\taskdist}{1.2cm}
\tikzstyle{eventline}=[thick,densely dashed,-]
\tikzstyle{monstate}=[ultra thick]
\newcommand{\monstateh}{.6}

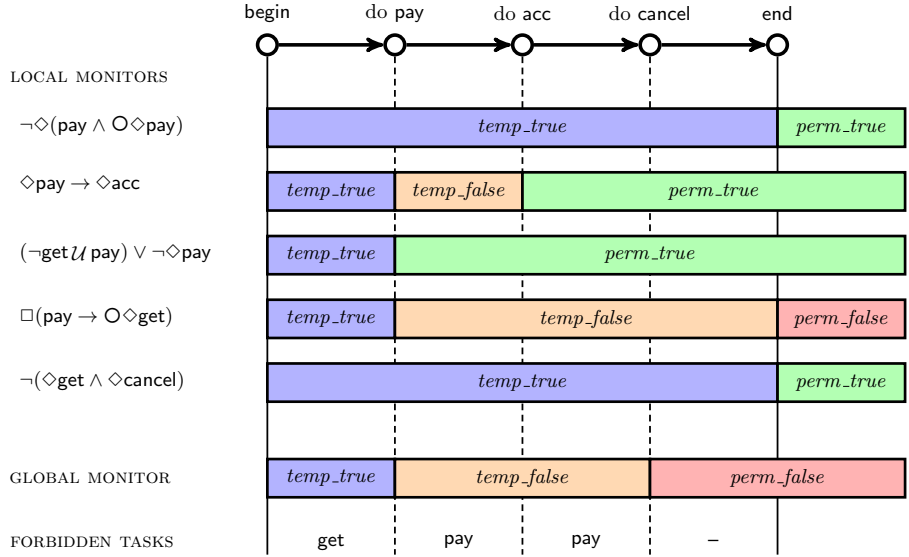
\begin{figure}[t!]
   \centering
%
\resizebox{\textwidth}{!}{
  \begin{tikzpicture}[node distance=\taskdist,y=1cm,x=2cm]

    \node[anchor=west,xshift=-1.5mm] at (0,.5) {\textsc{local monitors}};

    \node[anchor=north west] (pay) at (0,0) {\absencepay};

    \node[,anchor=north west] (pay) at (0,-1) {\respondedexistencepayacc};

    \node[anchor=north west] (pay) at (0,-2) {\precedencepayget
    };

     \node[anchor=north west] (pay) at (0,-3) {\responsepayget
    };

    \node[anchor=north west] (get) at (0,-4) {\notcoexistencegetcancel};

    \node[anchor=west,xshift=-1.6mm] at (0,-5.8) {\textsc{global monitor}};

    \node[dot] (start) at (2,1) {};
    \node at (2,1.5) {\activity{begin}};
    \node[dot] (e1) at (3,1) {};
    \node at (3,1.5) {do \activity{pay}};
    \node[dot] (e2) at (4,1) {};
    \node at (4,1.5) {do \activity{acc}};
    \node[dot] (e3) at (5,1) {};
    \node at (5,1.5) {do \activity{cancel}};
    \node[dot] (end) at (6,1) {};
    \node at (6,1.5) {\activity{end}};
    \path[-stealth',ultra thick]
      (start) edge (e1)
      (e1) edge (e2)
      (e2) edge (e3)
      (e3) edge (end);

    \draw[eventline,solid] (start) edge (2,-7);
    \draw[eventline] (e1) edge (3,-7);
    \draw[eventline] (e2) edge (4,-7);
    \draw[eventline] (e3) edge (5,-7);
    \draw[eventline,solid] (end) edge (6,-7);

    \node[temptruemonstate,fit={(2,0) ($(6,-\monstateh)$)}] {};
    \node[truemonstate,fit={(6,0) ($(7,-\monstateh)$)}] {};

    \node[temptruemonstate,fit={(2,-1) ($(3,-1-\monstateh)$)}] {};
    \node[tempfalsemonstate,fit={(3,-1) ($(4,-1-\monstateh)$)}] {};
    \node[truemonstate,fit={(4,-1) ($(7,-1-\monstateh)$)}] {};

    \node[temptruemonstate,fit={(2,-2) ($(3,-2-\monstateh)$)}] {};
    \node[truemonstate,fit={(3,-2) ($(7,-2-\monstateh)$)}] {};

    \node[temptruemonstate,fit={(2,-3) ($(3,-3-\monstateh)$)}] {};
    \node[tempfalsemonstate,fit={(3,-3) ($(6,-3-\monstateh)$)}] {};
    \node[falsemonstate,fit={(6,-3) ($(7,-3-\monstateh)$)}] {};

    \node[temptruemonstate,fit={(2,-4) ($(6,-4-\monstateh)$)}] {};
    \node[truemonstate,fit={(6,-4) ($(7,-4-\monstateh)$)}] {};

    \node[temptruemonstate,fit={(2,-5.5) ($(3,-5.5-\monstateh)$)}] {};
    \node[tempfalsemonstate,fit={(3,-5.5) ($(5,-5.5-\monstateh)$)}] {};
    \node[falsemonstate,fit={(5,-5.5) ($(7,-5.5-\monstateh)$)}] {};

    \node[anchor=west,xshift=-1.5mm] at (0,-6.8) {\textsc{forbidden tasks}};

    \node at (2.5,-6.8) {\activity{get}};
    \node at (3.5,-6.8) {\activity{pay}};
    \node at (4.5,-6.8) {\activity{pay}};
    \node at (5.5,-6.8) {\activity{--}};

  \end{tikzpicture}
}
\caption{Result computed by monitoring the \declare model of Figure~\ref{fig:declare} against the noncompliant trace $\activity{pay}\cdot\activity{acc}\cdot\activity{cancel}$, considering local monitors for each constraint separately, and the global monitor accounting for all of them at once.\label{fig:monitoring}}
\end{figure}

\begin{example}
\label{ex:monitoring}
Figure~\ref{fig:monitoring} depicts the result computed by monitoring the \declare model introduced in Example~\ref{ex:declare}
and shown in Figure~\ref{fig:declare}, against a trace where a registration is paid, the corresponding regulation is accepted, and then the registration is canceled.

When monitoring starts, all local monitors are in state $\temptrue$, and so is the global monitor. Task \activity{get ticket} is forbidden, since according to the \constraint{precedence} constraint connecting that task to \activity{pay registration} (i.e., formula $\precedencepayget$), a previous execution of \activity{pay registration} is needed.
When the payment is executed:
\begin{compactitem}[$\bullet$]
  \item the local monitor for the \constraint{responded existence} constraint linking \activity{pay registration} to \activity{accept regulation} (i.e., formula \respondedexistencepayacc) moves to $\tempfalse$, because it requires acceptance of the regulation (which has not been done yet);
  \item the local monitor for the \constraint{precedence} constraint linking \activity{get ticket} to \activity{pay registration} (i.e., formula \precedencepayget) moves to $\permtrue$, enabling once and for all the possibility of executing \activity{get ticket};
  \item the local monitor for the \constraint{response} constraint linking \activity{pay registration} to \activity{get ticket} (i.e., formula \responsepayget) moves to $\tempfalse$, because its satisfaction now demands a consequent execution of the \activity{get ticket} task.
\end{compactitem}
The global monitor also moves to $\tempfalse$, since there are two tasks that must be executed to satisfy the \constraint{responded existence} and \constraint{response} constraints, and it is indeed possible to execute them without violating other constraints. At the same time, further payments are now forbidden, due to the \constraint{absence 2} constraint attached to the \activity{pay registration} task (i.e., formula \absencepay).

The consequent execution of \activity{accept regulation} turns the state of the \constraint{responded existence} constraint linking \activity{pay registration} to \activity{accept regulation} (i.e., formula \respondedexistencepayacc) to $\permtrue$: since the regulation has now been accepted, the constraint is satisfied and will stay so no matter how the execution is continued.

The most interesting transition is the one triggered by the consequent execution of the \activity{cancel registration} task. While this event does not trigger any state change in the local monitors, it actually induces a transition of the global monitor to the permanent, $\permfalse$ RV state. In fact, no continuation of the trace will be able to satisfy all constraints of the considered model. More specifically, the sequence of events received so far induces a so-called \emph{conflict} \cite{MWM12} for the \constraint{response} constraint linking \activity{pay registration} to \activity{get ticket} (i.e., formula \responsepayget), and the \constraint{not coexistence} constraint relating \activity{get ticket} and \activity{cancel registration} (i.e., formula \notcoexistencegetcancel). In fact, the \constraint{response} constraint requires a future execution of the \activity{get ticket} task, which is however forbidden by the \constraint{not coexistence} constraint. Consequently, no continuation of the current trace will satisfy both constraints at once.

Since no further task execution actually happens, the trace is finally declared to be complete, with no execution of the \activity{get ticket} task. This has the effect of respectively moving the \constraint{response} and \constraint{not coexistence} constraints to $\permfalse$ and $\permtrue$. Also the \constraint{absence 2} constraint on payment becomes $\permtrue$, witnessing that no double payment occurred in the trace.
\end{example}

\section{Modeling and monitoring metaconstraints}
\label{sec:monitoring-metaconstraints}
In Section~\ref{sec:rtm}, we have demonstrated that \LDLf has the ability of expressing formulae that capture the RV state of other formulae. This can be interpreted as the ability of \LDLf to express meta-level properties of \LDLf constraints within the logic itself. Such properties, which we call \emph{metaconstraints}, can, in turn, be themselves monitored using the automata-theoretic approach described in Section~\ref{sec:rtm}.

In this section we elaborate on this observation, discussing how metaconstraints can be built, and illustrating interesting metaconstraint patterns.

\subsection{Modeling Metaconstraints}
Theorem~\ref{thm:rv-ltl} shows that, for an arbitrary \LDLf formula $\varphi$, four \LDLf formulae can be automatically constructed to express whether $\varphi$ is in one of the four RV states. Consequently, given $s \in \set{\temptrue,\tempfalse,\true,\false}$, and an \LTLf/\LDLf formula $\varphi$, we can consider formulae of the form $\rvass{\varphi}{s}$ as special atoms of the logic itself.

Such special atoms are used to check whether a trace brings $\varphi$ in state $s$. However, they cannot be used to explicitly characterize which are the paths that lead $\varphi$ to RV state $s$, i.e., that make formula $\rvass{\varphi}{s}$ true. Such paths can be readily obtained by constructing the regular expression for language $\L(\rvass{\varphi}{s})$, which we denote as $re_{\rvass{\varphi}{s}}$. For example, $re_{\rvass{\varphi}{\permfalse}} = \L(\DIAM{\pref_{\lnot\varphi}}\Endt\land \lnot
\DIAM{\pref_{\varphi}}\Endt)$ describes all paths culminating in a permanent violation of $\varphi$.

With these notions at hand, we can build \LTLf/\LDLf metaconstraints as standard \LTLf/\LDLf formulae that include:
\begin{compactitem}[$\bullet$]
\item formulae of the form $\rvass{\varphi}{s}$ as atoms;
\item formulae of the form $re_{\rvass{\varphi}{s}}$ as path expressions.
\end{compactitem}
A metaconstraint is then translated back into a standard \LDLf formula by replacing each sub-formula of the form $\rvass{\varphi}{s}$ with its corresponding \LDLf formula according to Theorem~\ref{thm:rv-ltl}, and each sub-formula of the form $re_{\rvass{\varphi}{s}}$ with its corresponding regular expression. A direct (non-optimized) way to calculate the regular expression for $re_{\rvass{\varphi}{s}}$ is to construct the automaton for $\rvass{\varphi}{s}$, and then to fold this automaton back into a regular expression (using standard techniques).

\subsection{Some Relevant Metaconstraint Patterns}

We present three types of metaconstraints, demonstrating the sophistication and versatility of the resulting framework.

\medskip
\noindent
\emph{Contextualizing constraints.} This type of metaconstraint is used to express that a constraint must hold \emph{while} another constraint is in some RV state.
The latter constraint, together with the specified state, consequently provides a monitoring \emph{context} for the former, \emph{contextualized} constraint.

Let us specifically consider the case of a \emph{contextualized absence}, where, given a task $\activity{a}$, the contextualized constraint has the form $\Box\neg\activity{a}$, and the context is provided by an arbitrary constraint $\varphi$ being in a given RV state $s$. This is formalized as:
\begin{equation}
  \BOX{re_{\rvass{\varphi}{s}}} (\neg \activity{a} \lor \Endt)
  \label{eq:context}
\end{equation}
where $\Endt$ denotes the end of the trace, as defined in Section~\ref{sec:LTLf-LDLf}; this is needed since, in \LDLf, $\neg \activity{a}$ expresses that some task different than $\activity{a}$ is executed, while we also want to accept the case where no task is performed at all (and the trace completes).
The idea of formula \eqref{eq:context} is to relativize the unrestricted $\Box$ operator to all and only those paths leading to RV state $s$ for $\varphi$, which are, in turn, characterized by the regular expression $re_{\rvass{\varphi}{s}}$.

A monitor for formula \eqref{eq:context} returns $\temptrue$ either when $\varphi$ is not in state $s$, but may evolve into such a state, or when $\varphi$ is in state $s$. In the latter situation, by inspecting the monitor one can see that task $\activity{a}$ if forbidden; this also means that upon the execution of $\activity{a}$, the monitor evolves into $\permfalse$. Finally, the monitor returns $\permtrue$ if $\varphi$ is not  in state $s$, and cannot enter into state $s$ in the future, no matter how the trace is continued.

\begin{example}
  \label{ex:context}
  Consider the constraint model in Figure~\ref{fig:declare}. We now want to express that it is not possible to get the ticket after the payment is done, until the regulation is accepted (if it was accepted before, no restriction applies). This can be seen as a \emph{contextualized absence} constraint forbidding \activity{get ticket} when the \constraint{responded existence} that links \activity{pay registration} to \activity{accept regulation} (i.e., formula \respondedexistencepayacc) is \emph{temporarily violated}.  Formally, to encode this, we instantiate formula \eqref{eq:context} into:
  \[
    \BOX{re_{\rvass{\{\Diamond \activity{pay} \limp \Diamond \activity{acc}\}}{\tempfalse}}}
    (\neg \activity{get} \lor \Endt)
  \]
  which, in turn, expands into:
  \[
    \BOX{(\neg \activity{pay})^*;\activity{pay};(\neg\activity{acc})^*}
    (\neg \activity{get} \lor \Endt)
  \]
\end{example}

\medskip
\noindent
\emph{Compensation constraints.} In  general terms, compensation refers to a behavior that has to be enforced when the current execution reaches an unexpected/undesired state. In our setting, the undesired state triggering a compensation is the \emph{permanent violation} of a property that captures a desired behavior, which, in turn, triggers the fact that \emph{another} formula, capturing the compensating behavior, has to be satisfied. We call the first formula the \emph{default constraint}, and the second formula its \emph{compensating constraint}.

Let us consider the general case of a default \LDLf constraint $\varphi$, and a compensating \LDLf constraint $\psi$. By noticing that once a trace permanently violates a constraint, then every possible continuation still permanently violates that constraint, we capture the \emph{compensation} of $\varphi$ by $\psi$ as:
\begin{equation}
  \rvass{\varphi}{\permfalse} \limp \psi
  \label{eq:compensation}
\end{equation}
The intuitive interpretation of formula \eqref{eq:compensation} is that either $\varphi$ never enters into the $\permfalse$ RV state, or $\psi$ holds. No requirement is placed regarding \emph{when} $\psi$ should be monitored in case $\varphi$ gets permanently violated. In fact, the overall compensation formula \eqref{eq:compensation} gets temporarily/permanently satisfied even when the compensating constraint $\psi$ is temporarily/permanently satisfied \emph{before} the moment when the default constraint $\varphi$ gets permanently violated. This may sound counterintuitive, as it is usually intended that the compensating behavior has to be exhibited \emph{as a reaction} to the violation. We can capture this intuition by turning formula \ref{eq:compensation} into the following \emph{reactive compensation} formula:
\begin{equation}
  \rvass{\varphi}{\permfalse} \limp \DIAM{re_{\rvass{\varphi}{\permfalse}}}\psi
  \label{eq:compensation-refined}
\end{equation}
This formula imposes that, in case of a permanent violation of $\varphi$, the compensating constraint $\psi$ must hold \emph{after} $\varphi$ has become permanently violated.

Assuming that $\varphi$ can be potentially violated (which is the reason why we want to express a compensation), a monitor for formula \eqref{eq:compensation-refined} starts by emitting $\temptrue$. As soon as the monitored execution is so that $\varphi$ cannot be permanently violated anymore, the monitor switches to $\permtrue$. If instead the monitored execution leads to permanently violate $\varphi$, from the moment of the violation onwards, the evolution of the monitor follows that of $\psi$.

\begin{example}
\label{ex:compensation}
  Consider the \constraint{not coexistence} constraint in Figure~\ref{fig:declare}. We now want to model that, whenever this constraint is permanently violated, that is, whenever a ticket is retrieved and the registration is canceled, then a \activity{return ticket} (\activity{return} for short) task must  be executed. This has to occur in  reaction to the permanent violation. Hence, we rely on template \eqref{eq:compensation-refined} and instantiate it into:
  \[
      \rvass{\{\neg(\Diamond\activity{get}\land\Diamond\activity{cancel})\}}{\permfalse} \limp \DIAM{re_{\rvass{\{\neg(\Diamond\activity{get}\land\Diamond\activity{cancel})\}}{\false}}}\Diamond \activity{return}
  \]
  This formula is equivalent to
  \[
    (\Diamond\activity{get}\land\Diamond\activity{cancel})\limp \DIAM{re_{\{\Diamond\activity{get}\land\Diamond\activity{cancel}\}}}\Diamond \activity{return}
  \]
  which, in turn, becomes
\newcommand{\myregexp}{
\begin{array}{@{}l@{}l@{}}
&(o^*;\activity{get};(\neg \activity{cancel})^*;\activity{cancel};\true^*)\\
+&(o^*;\activity{cancel};\neg \activity{get}^*;\activity{get};\true^*)
\end{array}
}
  \[
    (\Diamond\activity{get}\land\Diamond\activity{cancel}) \limp
     \left\langle\myregexp\right\rangle\Diamond \activity{return}
  \]
  where $o$ is a shortcut notation for any task different than $\activity{get}$ and $\activity{cancel}$.
\end{example}

\medskip
\noindent
\emph{Constraint priority for conflict resolution.} Thanks to the fact that RV states take into considerations all possible future evolution of a monitored execution, our framework handles the subtle situation where the execution reaches a state of affairs in which the conjunction of two constraints is permanently violated, while none of the two is so if considered in isolation. This situation of \emph{conflict} has been already recalled in Section~\ref{sec:declare} in the case of \declare. A situation of conflict involving two constraints $\varphi$ and $\psi$ witnesses that even though none of $\varphi$ and $\psi$ is permanently violated, they contradict each other, and hence in every possible future course of execution at least one of them will eventually become permanently violated. In such a state of affairs, it may become relevant to specify which ones of the two constraints has \emph{priority} over the other, that is, which one should be preferably satisfied.

 Formally, a trace culminates in a conflict for two \LDLf constraints $\varphi$ and $\psi$ if it satisfies the following metaconstraint:
\begin{equation}
      \rvass{\{\varphi\land\psi\}}{\permfalse} \land \neg \rvass{\varphi}{\permfalse}
      \land \neg \rvass{\psi}{\permfalse}  \label{eq:conflict}
  \end{equation}
Specifically, assuming that $\varphi$ and $\psi$ can potentially enter into a conflict, a monitor for formula \eqref{eq:conflict} proceeds as follows:
\begin{compactitem}[$\bullet$]
\item Initially, the monitor outputs $\tempfalse$, witnessing that no conflict has been seen so far, but it may actually occur in the future.
\item From this initial situation, the monitor can evolve in one of the following two ways:
\begin{compactitem}[$-$]
\item the monitor turns to $\permfalse$, witnessing that from this moment on neither of the two constraints will ever be violated anymore, irrespectively of how the trace continues;
\item the monitor turns to $\temptrue$, whenever the monitored execution indeed culminates in a conflict -- this witnesses that a conflict is currently in place.
\end{compactitem}
\item From the latter situation witnessing the presence of a conflict, the monitor evolves then to $\permfalse$ when one of the two constraints indeed becomes permanently violated; this witnesses that the conflict is not anymore in place, due to the fact that now the permanent violation can actually be ascribed to one of the two constraints taken in isolation from the other.
\end{compactitem}
Using this monitor, we can identify all points in the trace where a conflict is in place by simply checking when the monitor returns $\temptrue$.

Notice that the monitor never outputs $\permtrue$, since a conflicting situation will always eventually permanently violate $\varphi$ or $\psi$, in turn, permanently violating \eqref{eq:conflict}. In addition, the notion of conflict defined in formula \eqref{eq:conflict} is inherently ``non-monotonic'', as it ceases to exist as soon as one of the two involved constraints becomes permanently violated alone.
This is the reason why we cannot directly employ formula \eqref{eq:conflict} as a basis to define which constraint we \emph{prefer} over the other when a conflict arises. To declare that $\varphi$ is \emph{preferred over} $\psi$, we then relax formula \eqref{eq:conflict} by simply considering the violation of the composite constraint $\varphi \land \psi$, which may occur due to a conflict or due to the permanent violation of one of the two constraints $\varphi$ and $\psi$.
We then create a formula expressing that whenever the composite constraint is violated, then we want to satisfy the preferred constraint $\varphi$:
 \begin{equation}
  \DIAM{re_{\rvass{\{\varphi\land\psi\}}{\permfalse}}} \ttrue \limp \varphi
  \label{eq:preference}
\end{equation}
This pattern can be generalized to conflicts involving $n$ formulae, using their proper maximal subsets as building blocks.
In the typical situation where a permanent violation of $\varphi \land \psi$ does not manifest itself at the beginning of the trace, but may indeed occur in the future, a monitor for \eqref{eq:preference} starts by emitting $\temptrue$. When the composite constraint $\varphi \land \psi$ becomes permanently violated (either because of a conflict, or because of a permanent violation of one of its components), formula $\rvass{\{\varphi\land\psi\}}{\permfalse}$ turns to $\permtrue$, and the monitor consequently switches to observe the evolution of $\varphi$ (that is, of the head of the implication in \eqref{eq:preference}).

\newcommand{\myconflictregexp}{(o^*;\activity{pay};(o+\activity{pay})^*;\activity{cancel};(\neg \activity{get})^*)+(o^*;\activity{cancel};(o+\activity{cancel})^*;\activity{pay};(\neg \activity{get})^*)}

\newcommand{\myrelaxedconflictregexp}{
\begin{array}{@{}l@{}l@{}}
&(o^*;\activity{pay};(\neg \activity{cancel})^*;\activity{cancel};(\true)^*)
\\
+&
(o^*;\activity{get};(\neg \activity{cancel})^*;\activity{cancel};(\true)^*)\\
+&
(o^*;\activity{cancel};(\activity{cancel}+o)^*;(\activity{get}+\activity{pay});(\true)^*)
\end{array}
}
\begin{example}
\label{ex:conflict}
  Consider again Figure~\ref{fig:declare}, and in particular the \constraint{response} and \constraint{not coexistence} constraints respectively linking \activity{pay registration} to \activity{get ticket}, and \activity{get ticket} to \activity{cancel registration}, which we compactly refer to as $\psi_r$ and $\varphi_{nc}$. These two constraints conflict when a registration is paid and canceled, but the ticket is not retrieved (this would indeed lead to a permanent violation of $\varphi_{nc}$ alone). Let $o$ denote any task that is different from $\activity{pay}$, $\activity{get}$, and $\activity{cancel}$. The traces that culminate in a conflict for $\psi_r$ and $\varphi_{nc}$ are those that satisfy the regular expression:
\begin{equation}
\myconflictregexp
\label{eq:myconflict}
\end{equation}
Recall that, as specified in Section~\ref{sec:LTLf-LDLf}, testing whether a trace satisfies this regular expression can be done by encoding it in \LDLf as:
\begin{equation}
\DIAM{\myconflictregexp} \Endt
\label{eq:myconflict-check}
\end{equation}

We want now express that we prefer the \constraint{not coexistence} constraint over the \constraint{response} one, i.e., that, upon cancelation, the ticket should not be retrieved even if the payment has been done. To this end, we first notice that, for an evolving trace, the composite constraint $\psi_r \land \varphi_{nc}$ is permanently violated either when $\varphi_{nc}$ is so, or when a conflict arise. The first situation arises when the trace contains both $\activity{cancel}$ and $\activity{get}$ (in whatever order), whereas the second arises when the trace contains both $\activity{cancel}$ and $\activity{pay}$ (in whatever order). Consequently, we have that $re_{\rvass{\{\varphi_{nc}\land\psi_r\}}{\permfalse}}$ corresponds to the regular expression:
\[\myrelaxedconflictregexp\]
We then use this regular expression together with $\varphi_{nc}$ to instantiate formula \eqref{eq:preference} as follows:
\[
  \left\langle\myrelaxedconflictregexp\right\rangle\ttrue
   \limp \notcoexistencegetcancel
\]
\end{example}
We conclude by showing the evolution of the monitors for the metaconstraints discussed in the various examples of this section.

\renewcommand{\taskdist}{1.2cm}
\tikzstyle{eventline}=[thick,densely dashed,-]
\tikzstyle{monstate}=[ultra thick]
\renewcommand{\monstateh}{.6}

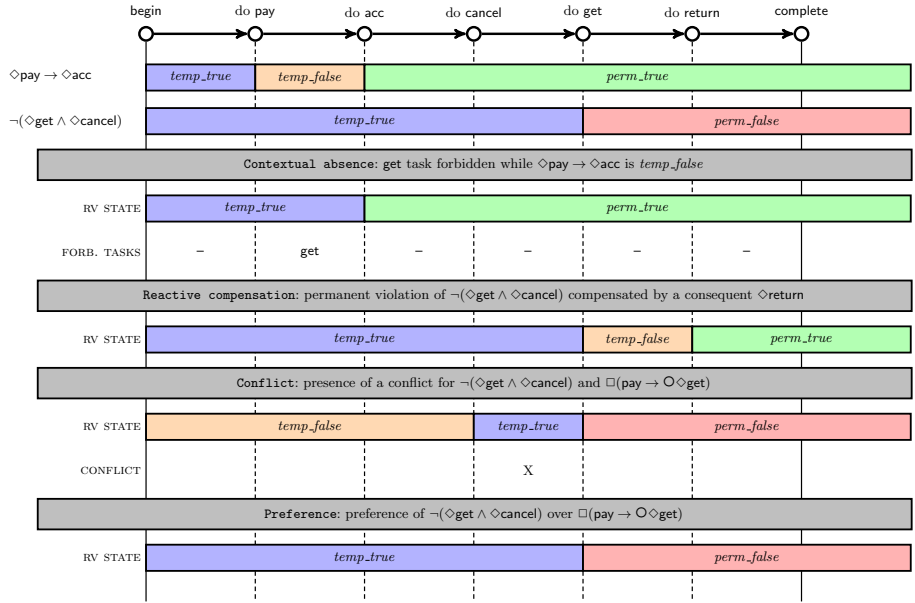
\begin{figure}[t!]
   \centering
%
\resizebox{\textwidth}{!}{
  \begin{tikzpicture}[node distance=\taskdist,y=1cm,x=2.5cm]
    \node[dot] (start) at (0,1) {};
    \node at (0,1.5) {\activity{begin}};
    \node[dot] (e1) at (1,1) {};
    \node at (1,1.5) {do \activity{pay}};
    \node[dot] (e2) at (2,1) {};
    \node at (2,1.5) {do \activity{acc}};
    \node[dot] (e3) at (3,1) {};
    \node at (3,1.5) {do \activity{cancel}};
    \node[dot] (e4) at (4,1) {};
    \node at (4,1.5) {do \activity{get}};
    \node[dot] (e5) at (5,1) {};
    \node at (5,1.5) {do \activity{return}};
    \node[dot] (end) at (6,1) {};
    \node at (6,1.5) {\activity{complete}};
    \path[-stealth',ultra thick]
      (start) edge (e1)
      (e1) edge (e2)
      (e2) edge (e3)
      (e3) edge (e4)
      (e4) edge (e5)
      (e5) edge (end);

    \draw[eventline,solid] (start) edge (0,-12);
    \draw[eventline] (e1) edge (1,-12);
    \draw[eventline] (e2) edge (2,-12);
    \draw[eventline] (e3) edge (3,-12);
    \draw[eventline] (e4) edge (4,-12);
    \draw[eventline] (e5) edge (5,-12);
    \draw[eventline,solid] (end) edge (6,-12);

  \node[anchor=west] at (-1.3,0) {\respondedexistencepayacc};
   \node[temptruemonstate,fit={(0,.5*\monstateh) (1,-.5*\monstateh)}] {};
   \node[tempfalsemonstate,fit={(1,.5*\monstateh) (2,-.5*\monstateh)}] {};
    \node[truemonstate,fit={(2,.5*\monstateh) (7,-.5*\monstateh)}] {};

  \node[anchor=west] at (-1.3,-1) {\notcoexistencegetcancel};
   \node[temptruemonstate,fit={(0,-1+.5*\monstateh) (4,-1-.5*\monstateh)}] {};
   \node[falsemonstate,fit={(4,-1+.5*\monstateh) (7,-1-.5*\monstateh)}] {};

  \node[fill=lightgray,minimum width=8*2.5cm,draw,very thick,anchor=west,minimum height=7mm] at (-1,-2){};
   \node[anchor=center] at (3,-2) {\constraint{Contextual absence}: $\activity{get}$ task forbidden while $\respondedexistencepayacc$ is $\tempfalse$
    };

  \node[anchor=east] at (0,-3) {\textsc{rv state}};
  \node[temptruemonstate,fit={(0,-3+.5*\monstateh) (2,-3-.5*\monstateh)}] {};
  \node[truemonstate,fit={(2,-3+.5*\monstateh) (7,-3-.5*\monstateh)}] {};

   \node[anchor=east] at (0,-4) {\textsc{forb.~tasks}};
    \node at (.5,-4) {\activity{--}};
    \node at (1.5,-4) {\activity{get}};
    \node at (2.5,-4) {\activity{--}};
    \node at (3.5,-4) {\activity{--}};
    \node at (4.5,-4) {\activity{--}};
    \node at (5.5,-4) {\activity{--}};

  \node[fill=lightgray,minimum width=8*2.5cm,draw,very thick,anchor=west,minimum height=7mm] at (-1,-5){};
   \node[anchor=center] at (3,-5) {
      \constraint{Reactive compensation}: permanent violation of $\notcoexistencegetcancel$ compensated by a consequent $\Diamond \activity{return}$
    };
  \node[anchor=east] at (0,-6) {\textsc{rv state}};
  \node[temptruemonstate,fit={(0,-6+.5*\monstateh) (4,-6-.5*\monstateh)}] {};
  \node[tempfalsemonstate,fit={(4,-6+.5*\monstateh) (5,-6-.5*\monstateh)}] {};
  \node[truemonstate,fit={(5,-6+.5*\monstateh) (7,-6-.5*\monstateh)}] {};

  \node[fill=lightgray,minimum width=8*2.5cm,draw,very thick,anchor=west,minimum height=7mm] at (-1,-7){};
   \node[anchor=center] at (3,-7) {
      \constraint{Conflict}: presence of a conflict for $\notcoexistencegetcancel$ and $\responsepayget$
    };
\node[anchor=east] at (0,-8) {\textsc{rv state}};
  \node[tempfalsemonstate,fit={(0,-8+.5*\monstateh) (3,-8-.5*\monstateh)}] {};
  \node[temptruemonstate,fit={(3,-8+.5*\monstateh) (4,-8-.5*\monstateh)}] {};
  \node[falsemonstate,fit={(4,-8+.5*\monstateh) (7,-8-.5*\monstateh)}] {};

  \node[anchor=east] at (0,-9) {\textsc{conflict}};
    \node at (.5,-9) {};
    \node at (1.5,-9) {};
    \node at (2.5,-9) {};
    \node at (3.5,-9) {X};
    \node at (4.5,-9) {};
    \node at (5.5,-9) {};

  \node[fill=lightgray,minimum width=8*2.5cm,draw,very thick,anchor=west,minimum height=7mm] at (-1,-10){};
   \node[anchor=center] at (3,-10) {
      \constraint{Preference}: preference of $\notcoexistencegetcancel$ over $\responsepayget$
    };
  \node[anchor=east] at (0,-11) {\textsc{rv state}};
  \node[temptruemonstate,fit={(0,-11+.5*\monstateh) (4,-11-.5*\monstateh)}] {};
  \node[falsemonstate,fit={(4,-11+.5*\monstateh) (7,-11-.5*\monstateh)}] {};

  \end{tikzpicture}
}
\caption{Result computed by monitoring the metaconstraints in Examples~\ref{ex:context}, \ref{ex:compensation}, and \ref{ex:conflict} against the trace $\activity{pay}\cdot\activity{acc}\cdot\activity{cancel}\cdot\activity{get}\cdot\activity{return}$; for readability, we also report the evolution of the monitors for the constraints mentioned by the metaconstraints. \label{fig:monitoring-metaconstraints}}
\end{figure}

\begin{example}
  Figure~\ref{fig:monitoring-metaconstraints} reports the result computed by the monitors for the metaconstraints discussed in Examples~\ref{ex:context}, \ref{ex:compensation}, and \ref{ex:conflict} on a sample trace. When the payment occurs, the \constraint{contextual absence} constraint forbids to get tickets. The prohibition is then permanently removed upon the consequent acceptance of the regulation, which ensures that the selected context will never appear again.

  The execution of the third step, consisting in the cancelation of the order, induces a conflict for $\notcoexistencegetcancel$ and $\responsepayget$, since they respectively forbid and require to eventually  get the ticket. The monitor for the \constraint{conflict} metaconstraint witnesses this by switching to $\temptrue$. The \constraint{preference} stays instead $\temptrue$, but while up to this point it was emitting $\temptrue$ because no conflict had occurred yet, it now emits $\temptrue$ because this is the current RV state of the preferred, \constraint{not coexistence} constraint.

  The execution of the \activity{get ticket} task induces a permanent violation for constraint $\notcoexistencegetcancel$, which, in turn, triggers a number of effects:
  \begin{compactitem}
  \item Since the \constraint{preference} metaconstraint is now following the evolution of the preferred constraint $\notcoexistencegetcancel$, it also moves to $\permfalse$.
  \item The conflict is not present anymore and will never be encountered again, given that one of its two constraints is permanently violated on its own. Thus, the monitor for the \constraint{conflict} metaconstraint turns to $\permfalse$.
  \item The \constraint{reactive compensation} is triggered by the permanent violation of $\notcoexistencegetcancel$, and asserts that, from now on, the compensating constraint $\Diamond \activity{return}$ must be satisfied; since the ticket is yet to be returned, the metaconstraint turns to $\tempfalse$.
  \end{compactitem}

  The execution of the last step, consisting in returning the ticket, has the effect of permanently satisfying the \constraint{compensation} metaconstraint, which was indeed waiting for this task to occur.
  \end{example}

\section{Implementation} \label{sec:implementation}
The entire approach has been implemented as an \emph{operational decision support (OS) provider} for the \textsc{ProM} 6 process mining
framework\footnote{\url{http://www.promtools.org/prom6/}} called \textsc{LDL Monitor}. \textsc{ProM} 6 provides a generic OS environment~\cite{Westergaard2011:OS,Maggi2017} that supports the interaction between an external workflow management system at runtime (producing events) and \textsc{ProM}. In Section~\ref{general}, we will sketch some relevant aspects of the general architecture of the OS backbone implemented inside ProM 6. In Section~\ref{interaction}, we ground the discussion to the specific case
of the \textsc{LDL Monitor}, discussing the skeleton of our compliance verification OS Provider. The data exchanged between the \textsc{LDL Monitor} client and provider
is illustrated in Section~\ref{data}. In Section~\ref{client}, we describe the implemented \textsc{LDL Monitor} client. At the back-end of the \textsc{LDL Monitor}, there is a software module
specifically dedicated to the construction and manipulation of {\NFA}s
from \LDLf/\LTLf formulae (detailed in Section~\ref{sec:backend}), concretely implementing the technique presented
in Section~\ref{sec:automaton}. This software is called FLLOAT, which stands for ``From \LTLf/\LDLf To AuTomata'', its code is open source and publicly available at \url{https://github.com/RiccardoDeMasellis/FLLOAT}.

\begin{figure}[t]
\centering
\includegraphics[width = 1.\textwidth]{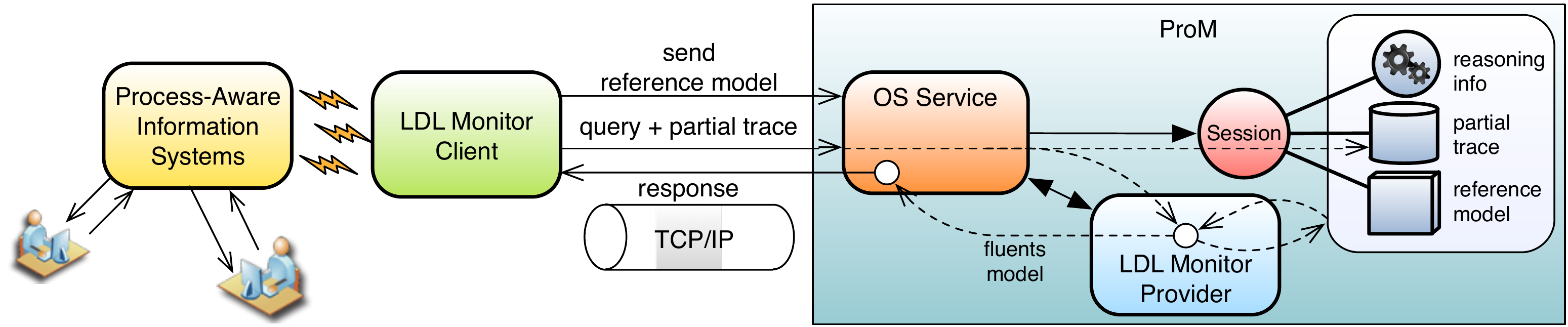}
\caption{ProM OS backbone architecture.}
\vspace{-0.3cm}
\label{archi}
\end{figure}




\subsection{General Architecture}
\label{general}
The ProM OS architecture (shown in Figure~\ref{archi}) relies on the well-known client-server paradigm \cite{DBLP:conf/fase/MaggiMA12}.
More specifically, the ProM OS Service manages the
interaction with running process instances and acts as a mediator between them
and the registered specific OS providers.
Sessions are created and handled by the OS Service to maintain the state of the interaction with each running client.
To establish a stateful connection with
the OS Service, the client creates a session handle for each managed
running process instance, by providing host
and port of the OS Service.
When the client sends a first query related to one of such running
instances to the OS Service, it specifies information related to the
initialization of the connection (such as reference models,
configuration parameters, etc.) and to the type of the queries that
will be asked during the execution. This latter
information will be used by the OS Service to select, among the
registered active providers, the ones that can answer the received
query. The session handle takes care of the interaction with the
service from the client point of view, hiding the connection details
and managing the information passing in a lazy way. The interaction
between the handle and the service takes place over a TCP/IP connection.

\subsection{\textsc{LDL Monitor} Skeleton}
\label{interaction}
In the \textsc{LDL Monitor}, the interaction between a client and the OS
Service mainly consists of two aspects.
First of all, before starting the runtime compliance verification task, the
  client sends to the OS Service the \LDLf reference model to be used.
  This model is then placed inside the session by the OS service. The
  reference model is a set of \LDLf constraints represented as strings.
The client can also set further information and properties. For
example, each constraint in the \LDLf reference model can be
associated to a specific weight, that can be then exploited to compute
metrics and indicators that measure the degree
of adherence of the running instance to the reference model.
Secondly, during the execution, the client sends queries about the
  current monitoring status for one of the managed process instances. The session handle augments these queries
with the partial execution trace containing the
  evolution that has taken place for the process instance after the last request.
The OS Service handles a query by first storing the
events received from the client, and then invoking the \textsc{LDL Monitor}
provider.
The \textsc{LDL Monitor} provider recognizes whether it is being invoked
for the first time with respect to~that process instance. If this is the case,
it takes care of translating the reference model onto the underlying
formal representation.
The provider then returns a fresh result to the client, exploiting a
reasoning component for the actual result's computation. The reasoning component behind the provider is described in Section \ref{sec:backend}.
After each query, the generated result is sent back to the OS Service, which possibly
combines it with the results produced by other relevant providers,
finally sending the global response back to the client.

\begin{figure}[t]
\centering
\includegraphics[width = 0.89\textwidth]{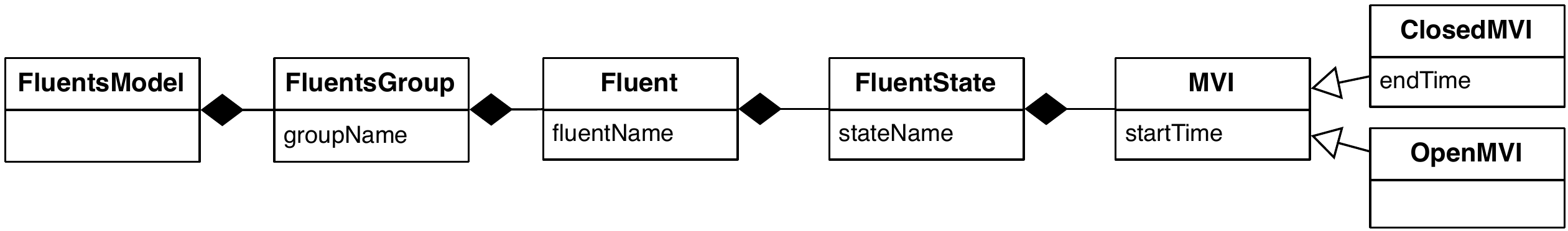}
\caption{Fluent model used to store the evolution of constraints.\label{fig:fluentsModel}\vspace{-0.3cm}}
\end{figure}

\subsection{Exchanged Data and Business Constraints States}
\label{data}
We now discuss the data exchanged by the \textsc{LDL Monitor} client and
provider.
The partial execution traces sent by the client to the OS use the
XES format (\url{www.xes-standard.org/}) for event data.
XES is an extensible XML-based standard recently adopted
by the IEEE task force on process mining.
The response produced by the \textsc{LDL Monitor} provider is composed of
two parts. The first part contains the temporal information related to the
evolution of each monitored business constraint from the beginning of
the trace up to now. At each time point, a constraint can be in one
state, which models whether it is currently: \emph{(permanently) satisfied}, i.e., the current execution trace complies with the
  constraint; \emph{possibly satisfied}, i.e., the current execution trace is compliant with the constraint, but it is possible to violate it
  in the future;
\emph{(permanently) violated}, i.e., the process instance is not compliant with
  the constraint; \emph{possibly violated}, i.e., the current execution trace is
  not compliant with the constraint, but it is possible to satisfy it
  by generating some sequence of events.
This state-based evolution is encapsulated in a \emph{fluent model}
which obeys to the schema sketched in Figure~\ref{fig:fluentsModel}.
A fluent model aggregates fluents groups, containing
sets of correlated fluents. Each fluent models a multi-state property that changes
over time. In our setting, fluent names refer to the
constraints of the reference model. The fact that the constraint was
in a certain state along a (maximal) time interval is modeled by
associating a closed MVI (Maximal Validity Interval) to that state. MVIs
are characterized by their starting and ending
timestamps. Current states are associate to open MVIs, which have an initial fixed timestamp but an
end that will be bounded to a currently unknown future
value.

\begin{figure}[t!]
\centering
\includegraphics[scale=0.7]{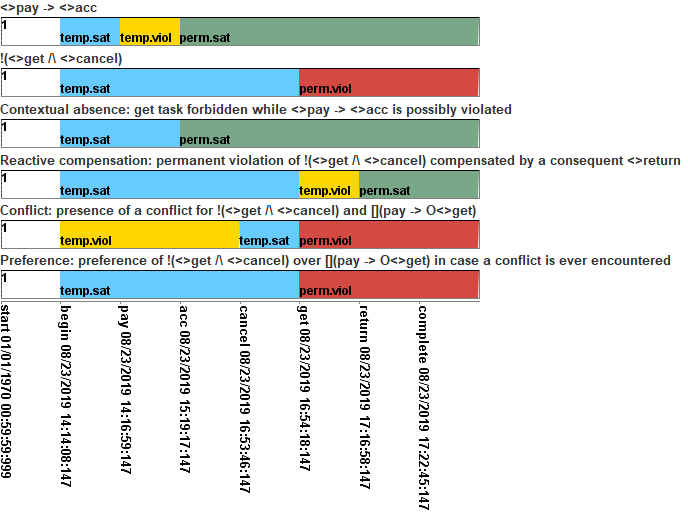}
\caption{Screenshot of one of the \textsc{LDL Monitor} clients.\label{fig:client}\vspace{-0.3cm}}
\end{figure}

\subsection{LDL Monitor Client}
\label{client}
We have developed two \textsc{LDL Monitor} clients, in order to deal with different settings:
(a) replay of a process instance starting from a complete event log, and (b) acquisition of events from an information system.
The first client is mainly used for testing and experimentation.
The second client requires a connection to some information system, e.g., a workflow management system.
The two clients differ on how the user is going to provide the stream of events,
but they both include an interface with a graphical representation of the obtained fluent model,
showing the evolution of constraints and also reporting the trend of the compliance indicator. Figure~\ref{fig:client} shows the interface running with the example in Figure~\ref{fig:monitoring-metaconstraints}.


\input{7.1-implementation-backend}
\section{Conclusions}
In this article, we have brought forward a foundational and practical approach to formalize and monitor linear temporal constraints and metaconstraints, under the assumption that the traces generated by the system under study are \emph{finite}. This is, e.g., the typical case in the context of business process management and service-oriented architectures, where each execution of a business process or service invocation leads from a starting state to a completion state in a possibly unbounded, yet finite, number of steps.

The main novelty of our approach is to adopt a more powerful specification logic, that is, \LDLf (which corresponds to Monadic Second-Order Logic over finite traces), instead of the typical choice of \LTLf (which corresponds to First-Order Logic over finite traces). Like in the case of \LTLf, also \LDLf comes with an automata-theoretic characterization that employs standard finite-state automata. Differently from \LTLf, though, \LDLf can declaratively express, within the logic, not only constraints that predicate on the dynamics of task executions, but also constraints that predicate on the monitoring state of other constraints.

The approach has been fully implemented as an independent library to specify \LDLf/\LTLf formulae as well as obtain and manipulate their corresponding automata, which is then invoked by a process monitoring infrastructure that has been developed within the state-of-the-art ProM process mining framework.

As a next step, we intend to incorporate other monitoring perspectives, such as the data perspective taking into consideration the data carried by the monitored events. This setting is reminiscent of stream query languages and event calculi. For example, the logic-based Event Calculus has been applied to process monitoring against data-aware extensions of the \declare language in \cite{MMC13}, also considering some specific forms of compensation \cite{CMM08}. However, all these approaches are only meant to query and reason over (a portion of) the events collected so far in a trace, and not to reason upon its possible future continuations, as we do in our approach. Genuine investigation is then required towards understanding under which conditions it is possible to lift the automata-based techniques presented here to the case where events are equipped with a data payload and constraints are expressed in (fragments of) first-order temporal logics over finite traces.

\bibliographystyle{ACM-Reference-Format-Journals}
\bibliography{main-bib}

\begin{thebibliography}{10}
\providecommand{\url}[1]{\texttt{#1}}
\providecommand{\urlprefix}{URL }
\providecommand{\doi}[1]{https://doi.org/#1}

\bibitem{Aal16}
van~der Aalst, W.M.P.: Process Mining - Data Science in Action, Second Edition.
  Springer (2016). \doi{10.1007/978-3-662-49851-4},
  \url{https://doi.org/10.1007/978-3-662-49851-4}

\bibitem{Bauer2010:LTL}
Bauer, A., Leucker, M., Schallhart, C.: {Comparing LTL Semantics for Runtime
  Verification}. Logic and Computation  (2010)

\bibitem{BrDP18}
Brafman, R.I., {De Giacomo}, G., Patrizi, F.: Ltlf/ldlf non-markovian rewards.
  In: {AAAI}. pp. 1771--1778. {AAAI} Press (2018)

\bibitem{CDLV09}
Calvanese, D., De~Giacomo, G., Lenzerini, M., Vardi, M.Y.: An
  automata-theoretic approach to {R}egular {XPath}. LNCS, vol.~5708, pp.
  18--35. Springer (2009)

\bibitem{CMM08}
Chesani, F., Mello, P., Montali, M., Torroni, P.: Verification of
  choreographies during execution using the reactive event calculus. In: 5th
  Int. Workshop on Web Services and Formal Methods (WS-FM). LNCS, vol.~5387.
  Springer (2008)

\bibitem{ClDe99b}
Clark, J., DeRose, S.: {XML} {P}ath {L}anguage ({XPath}) version 1.0. {W3C}
  {R}ecommendation, World Wide Web Consortium (Nov 1999)

\bibitem{DDGM14}
{De Giacomo}, G., {De Masellis}, R., Grasso, M., Maggi, F.M., Montali, M.:
  Monitoring business metaconstraints based on {LTL} and {LDL} for finite
  traces. LNCS, vol.~8659. Springer (2014)

\bibitem{DegVa13}
De~Giacomo, G., Vardi, M.Y.: Linear temporal logic and linear dynamic logic on
  finite traces. In: Proc.\ of the 23rd Int.\ Joint Conf.\ on Artificial
  Intelligence (IJCAI). AAAI (2013)

\bibitem{DwAC99}
Dwyer, M.B., Avrunin, G.S., Corbett, J.C.: Patterns in property specifications
  for finite-state verification. In: Boehm, B.W., Garlan, D., Kramer, J. (eds.)
  Proc. of the 1999 International Conf. on Software Engineering (ICSE). ACM
  Press (1999)

\bibitem{FiLa79}
Fischer, M.J., Ladner, R.E.: Propositional dynamic logic of regular programs.
  Journal of Computer and System Science  \textbf{18} (1979)

\bibitem{HaKT00}
Harel, D., Kozen, D., Tiuryn, J.: Dynamic Logic. MIT Press (2000)

\bibitem{LMM13}
Ly, L.T., Maggi, F.M., Montali, M., Rinderle-Ma, S., van~der Aalst, W.M.P.: A
  framework for the systematic comparison and evaluation of compliance
  monitoring approaches. In: Proc. of the 17th IEEE Int. Enterprise Distributed
  Object Computing Conf. (EDOC). IEEE (2013)

\bibitem{DBLP:conf/fase/MaggiMA12}
Maggi, F.M., Montali, M., van~der Aalst, W.M.P.: An operational decision
  support framework for monitoring business constraints. In: Fundamental
  Approaches to Software Engineering - 15th International Conference, {FASE}
  2012, Held as Part of the European Joint Conferences on Theory and Practice
  of Software, {ETAPS} 2012, Tallinn, Estonia, March 24 - April 1, 2012.
  Proceedings. pp. 146--162 (2012)

\bibitem{MMW11}
Maggi, F.M., Montali, M., Westergaard, M., van~der Aalst, W.M.P.: Monitoring
  business constraints with linear temporal logic: An approach based on colored
  automata. In: Proc.\ of the 9th Int.\ Conf.\ on Business Process Management
  (BPM). LNCS, vol.~6896. Springer (2011)

\bibitem{Maggi2017}
Maggi, F.M., Westergaard, M.: Designing software for operational decision
  support through coloured petri nets. Enterprise {IS}  \textbf{11}(5),
  576--596 (2017)

\bibitem{MWM12}
Maggi, F.M., Westergaard, M., Montali, M., van~der Aalst, W.M.P.: Runtime
  verification of ltl-based declarative process models. In: Proc.\ of the 2nd
  Int.\ Conf.\ on Runtime Verification (RV). LNCS, vol.~7186. Springer (2012)

\bibitem{Marx04b}
Marx, M.: {XPath} with conditional axis relations. LNCS, vol.~2992, pp.
  477--494. Springer (2004)

\bibitem{Mon09}
Montali, M.: Specification and Verification of Declarative Open Interaction
  Models: a Logic-Based Approach. Ph.D. thesis, Department of Electronics,
  Computer Science and Telecommunications Engineering, University of Bologna
  (2009), \url{http://amsdottorato.cib.unibo.it/1829/}

\bibitem{Mon10}
Montali, M.: {Specification and Verification of Declarative Open Interaction
  Models: a Logic-Based Approach}, LNBIP, vol.~56. Springer (2010)

\bibitem{MMC13}
Montali, M., Maggi, F.M., Chesani, F., Mello, P., van~der Aalst, W.M.P.:
  Monitoring business constraints with the event calculus. ACM Trans.\ on
  Intelligent Systems and Technology  \textbf{5}(1) (2013)

\bibitem{MPVC10}
Montali, M., Pesic, M., van~der Aalst, W.M.P., Chesani, F., Mello, P., Storari,
  S.: Declarative specification and verification of service choreographies. ACM
  Trans.\ on the Web  \textbf{4}(1) (2010)

\bibitem{Pes08}
Pesic, M.: {Constraint-Based Workflow Management Systems: Shifting Controls to
  Users}. Ph.D. thesis, Beta Research School for Operations Management and
  Logistics, Eindhoven (2008)

\bibitem{PesV06}
Pesic, M., van~der Aalst, W.M.P.: A declarative approach for flexible business
  processes management. In: Proc. of the BPM 2006 Workshops. LNCS, vol.~4103.
  Springer (2006)

\bibitem{PeSV07}
Pesic, M., Schonenberg, H., van~der Aalst, W.M.P.: Declare: Full support for
  loosely-structured processes. In: Proc. of the 11th IEEE International
  Enterprise Distributed Object Computing Conf. (EDOC). pp. 287--300. IEEE
  Computer Society (2007)

\bibitem{Pnueli77}
Pnueli, A.: The temporal logic of programs. In: Proc.\ of the 18th Ann.\ Symp.\
  on Foundations of Computer Science (FOCS). IEEE (1977)

\bibitem{PrS96}
Prakken, H., Sergot, M.J.: Contrary-to-duty obligations. Studia Logica
  \textbf{57}(1) (1996)

\bibitem{Var11}
Vardi, M.: The rise and fall of linear time logic (2011),
  \url{http://www.cs.rice.edu/~vardi/papers/gandalf11-myv.pdf}, 2nd Int'l
  Symp.~on Games, Automata, Logics and Formal Verification

\bibitem{Westergaard2011:OS}
Westergaard, M., Maggi, F.: {Modelling and Verification of a Protocol for
  Operational Support using Coloured Petri Nets}. In: {Proc.\ of ATPN} (2011)

\end{thebibliography}

\received{February 2007}{March 2009}{June 2009}








\end{document}